\documentclass{lmcs}

\title{Separability and Non-Determinizability of WSTS}

\author[W.~Czerwi\'nski]{Wojciech Czerwi\'nski\lmcsorcid{0000-0002-6169-868X}}[a]

\author[E.~Keskin]{Eren Keskin\lmcsorcid{0009-0009-5621-6568}}[b]

\author[S.~Lasota]{S{\l}awomir Lasota\lmcsorcid{0000-0001-8674-4470}}[a]

\author[R.~Meyer]{Roland Meyer\lmcsorcid{0000-0001-8495-671X}}[b]

\author[S.~Muskalla]{Sebastian Muskalla\lmcsorcid{0000-0001-9195-7323}}[c]

\author[K~N~Kumar]{K Narayan Kumar}[d]

\author[P.~Saivasan]{Prakash Saivasan}[e]

\address{University of Warsaw, Poland}
\email{wczerwin@mimuw.edu.pl, sl@mimuw.edu.pl}

\address{TU Braunschweig, Germany}
\email{e.keskin@tu-braunschweig.de, roland.meyer@tu-braunschweig.de}

\address{Schwarz IT, Germany}
\email{sebastian@muskalla.saarland}

\address{Chennai Mathematical Institute and UMI ReLaX, India}
\email{kumar@cmi.ac.in}

\address{Institute of Mathematical Sciences, India}
\email{prakashs@imsc.res.in}

\usepackage{style}

\usepackage{configurable_figures}

\begin{document}

\begin{abstract}
    We study the languages recognized by well-structured transition systems (\WSTS) with upward and downward compatibility.
    Our first result shows that every pair of disjoint \WSTS languages is regularly separable:
    there is a regular language containing one of them while being disjoint from the other.
    As a consequence, if a language as well as its complement are both recognized by \WSTS, then they are necessarily regular.
    Our second result shows that the languages recognized by deterministic \WSTS form a strict subclass of the languages recognized by all \WSTS:
    we give a non-deterministic WSTS language that we prove cannot be recognized by a deterministic WSTS.
    The proof relies on a novel characterization of the languages accepted by deterministic WSTS.
\end{abstract}

\maketitle

\section{Introduction}

We study the languages recognized by well-structured transition systems~\cite{F87,Finkel90,AJ93,ACJT96,WSTSEverywhere01}.
\WSTS form a framework subsuming widely-studied models like Petri nets~\cite{Esparza1996} and their extensions with transfer~\cite{DufourdFinkelReset}, data~\cite{RV12}, and time~\cite{ADM04}, graph rewriting systems~\cite{JK2008}, depth-bounded systems~\cite{M2008,WZH10,EDOsualdo}, ad-hoc networks~\cite{ADRST11}, process algebras~\cite{BGZ04},
lossy channel systems~\cite{AJ93}, and programs running under weak memory models~\cite{ABBM10,ABBM12}.
Besides their generality, the importance of \WSTS stems from numerous decidability results. 
Finkel showed the decidability of termination and boundedness~\cite{F87,Finkel90}.
Abdulla came up with a backward algorithm for coverability~\cite{AJ93}, for which a matching forward procedure was only found much later~\cite{GEERAERTS2006180}.
Several simulation and equivalence problems are also decidable for \WSTS~\cite{WSTSEverywhere01}. 
The work on \WSTS\ even influenced algorithms for regular languages~\cite{Antichains2006}, and recently led to the study of new complexity classes~\cite{SchmitzSchnoebelen2011}.

A \WSTS\ is a transition system equipped with a quasi order on the states that satisfies two properties: it is a well quasi order (WQO) and it is (upward or downward) compatible with the transition relation in the sense that it forms a simulation relation~\cite{M71}.
For our language-theoretic study, we assume the transitions are labeled and the \WSTS is equipped with sets of initial and final states. 
The set of final states should be upward or downward closed wrt.~the quasi order of the \WSTS, a requirement that plays together well with the compatibility of the transition relation and that tends to hold in practice, as discussed above.  
When specialized to Petri nets or vector addition systems (VAS), this yields the well-known \emph{covering languages}. 
\paragraph*{Regular Separability}
For \WSTS languages, we study the problem of regular separability.
Given two languages $\alang$ and $\alangp$ over the same alphabet, a
\emph{separator} is a language $R$ that contains one of the languages and is disjoint from the other, $\alang \subseteq R$ and $R\cap \alangp = \emptyset$.
The separator is regular if it is a regular language.
Separability has recently attracted considerable attention.
We discuss the related work in a moment.
Disjointness is clearly necessary for regular separability.
We show that for {\WSTS}, disjointness is also sufficient.
Our first main result is this.
\begin{center}
Any two disjoint \WSTS languages are separable by a regular language.
\end{center}

The proof is in two steps.
First, we link regular separability to inductive invariants as known from verification~\cite{MP1995}. 
An inductive invariant is a set of states that contains the initial ones, is closed under the transition relation, and is disjoint from the final states.
We show that any inductive invariant (of a product of the given systems) gives rise to a regular separator~--~provided it can be finitely represented.
We do not even need \WSTS here, but only upward compatibility and that one of the systems is deterministic.

Second, we show that finitely-represented invariants always exist.
This task presents a challenge that may not be obvious. 
The proof technique established in the first step requires one of the systems to be deterministic. 
\begin{center}
Can WSTS be determinized (in a way that preserves the language)?
\end{center}
We will study this problem in greater detail in the second half of the paper, where we answer it negatively.
For the moment, observe that a canonical determinization with a Rabin and Scott-like Powerset construction~\cite{RabinScott59} has to fail, because WQOs are not stable under forming powersets \cite{JANCAR1999}.
Our way around the problem is the following.

We accept the fact that determinizing a WSTS no longer yields a WSTS, and carefully study the resulting class of transition systems.
They are formed over a lattice in which every infinite sequence has a subsequence that converges in a natural~sense.
This leads us to define the closure of an inductive invariant by adding the limits of all converging sequences inside the invariant. 
The key insight is that this closure is again an inductive invariant.
Together with the fact that closed sets have finitely many maximal elements, we arrive at the desired finite representation. 
%
%
%
%
We call the new transition systems converging, and believe they can be of interest beyond separability.

Our regular separability result has theoretical and practical implications. 
On the theoretical side, recall the following about Petri nets from~\cite{MukundKRS98asian,MukundKRS98icalp}: every two Petri net covering languages that are complements of each other are necessarily regular.
With our result, this actually holds for all instances of the \WSTS framework, and even across models as long as we stay in the realm of WSTS.  
%
For instance, if the covering language of a Petri net is the complement of the language of an lossy channel system, then they are necessarily regular;
and if the languages are just disjoint, they are regularly separable.

Our separability result is also important in verification, where the Software Verification Competition (SVCOMP) is the place for verification tools to compete and demonstrate their efficiency.
Since 2016, the \textsc{Ultimate} program analysis framework~\cite{ultimate} participates in SVCOMP and frequently wins medals across all categories, notably the overall gold medal.
The \textsc{Ultimate} framework implements language-theoretic algorithms~\cite{HHP10} that replace the classical state space search by a proof of language disjointness (between a refinement of the control-flow language and the language of undesirable behavior).
Regular separators are precisely what is needed to prove disjointness.

\paragraph*{Non-Determinizability}
As already alluded to, our second main result is that WSTS cannot be determinized in general.
We give a WSTS language $\witnesslang$ that we prove cannot be accepted by a deterministic WSTS.
%
%
The proof relies on a novel characterization of the deterministic \wsts\ languages that is also a contribution of ours: the deterministic \wsts\ languages are precisely the languages whose Nerode (right) quasi order is a WQO.
%
%
The characterization provides a first hint on how to construct $\witnesslang$.
The language should have an infinite antichain in the Nerode quasi order, for then this cannot be a WQO.
The second hint stems from easy to achieve determinizability results.
Finitely branching WSTS and WSTS over $\omega^2$-WQOs \emph{can} be determinized.
Hence, the WSTS accepting language $\witnesslang$ should be infinitely branching and the underlying WQO should be no $\omega^2$-WQO. 
%
%
The third hint can be found in~\cite[Section 2]{JANCAR1999}: WQOs that fail to be $\omega^2$-WQOs embed the so-called Rado WQO. 
Finally, the Rado WQO is known to have an infinite antichain when constructing downward-closed sets~\cite[Proposition 4.2]{IdealDecomp}.
%
%
%
The definition of $\witnesslang$ is thus guided by the idea of translating the Rado antichain into an antichain in the Nerode quasi order.
Interestingly, the underlying WSTS is deterministic except for the choice of the initial state.

\paragraph*{Compatibility}
We develop these results for upward-compatible WSTS~\cite{WSTSEverywhere01}.
To show that they also hold for downward-compatible WSTS, we establish general relationships between the models.
%
%
A key insight is that the complement of a deterministic upward-compatible WSTS is a deterministic downward-compatible WSTS. 
Moreover, the reversal of an upward-compatible WSTS language is a downward-compatible WSTS language.

\paragraph*{The Article}
This article combines the results of the papers \cite{CONCUR18,CONCUR23} to present an in-depth study of the WSTS languages.
The results in Section~\ref{Section:RegSepPrinciple} and Section~\ref{Section:SimpleRegSep}, as well as the determinizability results in Section~\ref{Section:NonDet}, stem from \cite{CONCUR18}.
The results in Sections~\ref{Section:CTS}, \ref{Section:FullRegSep}, \ref{Section:DWSTS}, along with the remaining results in Section \ref{Section:NonDet} stem from \cite{CONCUR23}.
Details and proofs missing in the former list of sections can be found in \cite{czerwinski2018sep}, while those that are missing in the latter list of sections can be found in \cite{keskinmeyer2023sep}.

\paragraph*{Related Work}
Separability is a  problem widely-studied in Theoretical Computer Science.
A classical result says that every two co-recursively enumerable languages are recursively separable, i.e.~separable by a recursive language~\cite{handbook}.
In the area of formal languages, separability \emph{of} regular languages by subclasses thereof
was investigated most extensively as a decision problem: given two regular languages, decide whether they are
separable by a language from a fixed subclass.
For the following subclasses, the separability problem of regular languages is decidable:
the piecewise-testable languages, shown independently in~\cite{CzerwinskiMM13} and~\cite{PlaceRZ13mfcs},
the locally testable and locally threshold-testable languages~\cite{PlaceRZ13fsttcs},
the languages definable in first-order logic~\cite{PlaceZ14lics},
and the languages of higher levels of the first-order hierarchy~\cite{PlaceZ14icalp}.

Regular separability of classes larger than the regular languages attracted little attention until recently. 
As an exception, the undecidability of regular separability for context-free languages has already been shown in the 70s~\cite{SzymanskiW76} (see also a later proof~\cite{Hunt82a}); then the undecidability has been strengthened to the visibly pushdown languages~\cite{Kopczynski16}, and to the languages of one-counter automata~\cite{CL17}.


The problem has also been studied in the context of Petri net/VAS languages.
With reachability as the acceptance condition, regular separability has been shown decidable only recently~\cite{KM24}, after it has been open for a decade. 
For subclasses of Petri net reachability languages, positive answers have been found earlier. 
Regular separability is $\PSPACE$-complete in the one-dimensional case~\cite{CL17}, 
and has elementary complexity for the languages recognizable by Parikh automata or, equivalently, $\mathbb{Z}$-VAS~\cite{CCLP17}.
Also for commutative closures of VAS languages, regular separability has been shown to be decidable~\cite{CCLP17b}.
As a consequence of our result, the regular separability problem for WSTS languages reduces to disjointness, and is thus trivially decidable.

Our regular separability result does not carry over to languages of infinite words.
Already for VAS, there are disjoint $\omega$-languages that cannot be separated by an $\omega$-regular language.
The algorithmic problem of checking whether a separator exists, however, is decidable~\cite{BMZ23}, but the algorithm that yields an \textsf{EXPSPACE} has to deal with a new form of non-linear systems of constraints~\cite{BKMZ24}.

The languages of upward-compatible \WSTS\ were also investigated in~\cite{WSTSLanguages07}, where interesting closure properties were shown, including a natural pumping lemma.
Subclasses of \WSTS languages were studied in~\cite{DELZANNO201312,ADB07,Martos-SalgadoR14}.

\subsection*{Outline}

We recall the basics on \WSTS in Section~\ref{Section:WSTS}.
In Section~\ref{Section:RegSepPrinciple}, we give the separability result.
The non-determinizability and determinizability results can be found in Section~\ref{Section:NonDet}.
Section~\ref{Section:DWSTS} generalizes these results to downward-compatible WSTS, and explores relations between the language classes.
Section~\ref{Section:Conclusion} concludes the paper.

\section{Well-Structured Transition Systems}\label{Section:WSTS}
We recall well-structured transition systems (WSTS) with upward compatibility~\cite{F87,AJ93,ACJT96,WSTSEverywhere01}.
Downward compatibility will be addressed in Section~\ref{Section:DWSTS}.
\paragraph*{Orders}
%
Let $(\states, \leq)$ be a quasi order and $P\subseteq \states$.
We call $P$ a chain if $\leq$ restricted to~$P$ is a total order.
We call $P$ an antichain if the elements in $P$ are pairwise incomparable.
The \emph{upward} and \emph{downward} closures of $P$ are defined as
\begin{align*}
    \upcls{P}\ =\ \setcond{\astatep\in\states}{\exists\astate\in P.\ \astate\leq\astatep}\qquad\downcls{P}\ =\ \setcond{\astatep\in\states}{\exists\astate\in P.\ \astate\geq\astatep}.
\end{align*}
We call $P$ upward closed, if $P=\upcls{P}$.
%
%
The powerset of $\states$ restricted to the upward-closed sets is $\uppowof{\states}$.
The downward-closed sets are defined similarly and we use $\downpowof{\states}$ for the downward-closed sets.
We write $\findownpowof{\states}$ to denote the downward closures of finite sets.


Let $(\states, \leq)$ be a partial order.
We write $\max P$ for the set of maximal elements in the subset $P\subseteq\states$.
They may not exist, in which case the set is empty.
We call $(\states, \leq)$ a complete lattice, if all $P\subseteq\states$ have a greatest lower bound in $\states$, also called meet and denoted by $\biglatmeet P \in\states$, and a least upper bound in $\states$, also called join and denoted by $\biglatjoin P\in\states$.
A~function $f:\states\rightarrow\states$ on a complete lattice is join preserving~\cite[Section 11.4]{C21}, if it distributes over arbitrary joins in that $f(\biglatjoin P)=\biglatjoin f(P)$ for all $P\subseteq \states$, where $f(P)=\setcond{f(p)}{p\in P}$. 
It is readily checked that join-preserving functions are monotonic: $p\leq q$ implies $f(p)\leq f(q)$. 
%
We call $(\states, \leq)$ a completely distributive lattice, if it is a complete lattice
where arbitrary meets distribute over arbitrary joins, and vice versa:
\begin{align*}
    \biglatmeet_{a\in A}\biglatjoin_{b\in B_{a}}\astate_{a, b}
    \ =\
    \biglatjoin_{f\in \mathit{C_{A, B}}}\biglatmeet_{a\in A}\astate_{a, f(a)}
    \qquad
    \biglatjoin_{a\in A}\biglatmeet_{b\in B_{a}}\astate_{a, b}
    \ =\
    \biglatmeet_{f\in \mathit{C_{A, B}}}\biglatjoin_{a\in A}\astate_{a, f(a)}\ .
\end{align*}
The definition makes use of the Axiom of Choice: $\mathit{C_{A, B}}$ denotes the set of choice functions that
map each $a\in A$ to a choice $b\in B_{a}$.
It is also important to note that, for any quasi order $(\states, \leq)$,
$(\downpowof{\states}, \subseteq)$ is a completely distributive lattice.
\paragraph*{Labeled Transition Systems}
A labeled transition system (LTS) $\anults=(\states, \initstates, \analph, \transitions, \finalstates)$
consists of a set of states $\states$, in our setting typically infinite,
a set of inital states $\initstates\subseteq\states$,
a set of final states $\finalstates\subseteq \states$,
a finite alphabet $\analph$, and
a set of labeled transitions
$\transitions: \states\times \analph\rightarrow \powof{\states}$.
The LTS is deterministic, if
$\sizeof{\initstates}=\sizeof{\transitions(\astate, \aletter)}=1$ for all
$\astate\in\states$ and all $\aletter\in\Sigma$.
It is finitely branching, if $\sizeof{\initstates}, \sizeof{\transitions(\astate, \aletter)}\in\nat$ for all $\astate\in\states$ and all $\aletter\in\Sigma$.

The language of the LTS is the set of words
that can reach a final state from an initial state:
\begin{align*}
    \langof{\anults}\ =\ \setcond{\aword\in\Sigma^{*}}{\apply{\initstates}{\aword}\cap\finalstates\neq\emptyset}\ .
\end{align*}
Here, we extend the transition relation to sets of states and languages:
\begin{align*}
    \apply{P}{\alang}\ =\ \bigcup_{\aword\in\alang}\transitions(P, \aword)\qquad
    \apply{P}{\aword\concat\aletter}\ =\ \apply{\apply{P}{\aword}}{\aletter}\qquad
    \apply{P}{\aletter}\ =\ \bigcup_{\astate\in P}\transitions(\astate, \aletter)\ .
\end{align*}
We write $\astate\trans{\aword} \astate'$ to denote $\astate'\in\apply{\astate}{\aword}$ if the set of transitions is obvious from the context.
We will often be interested in states that are reachable from an initial state resp.~states from which a final state can be reached:
\begin{align*}
    \reach{\anults}{\aword} \ &=\ \apply{\initstates}{\aword}\hspace{3em}
    \revreach{\anults}{\aword}\  =\ \invapply{\finalstates}{\aword}\\
    \reachall{\anults}\ &=\ \apply{\initstates}{\analph^{*}}\hspace{4.05em}
    \revreachall{\anults}\  =\ \invapply{\finalstates}{\analph^{*}}
\end{align*}

Finally, if the LTS is deterministic, we write $(\states, \initstate, \analph, \transitions, \finalstates)$ and $\transitions(\astate, \aletter)=\astatep$ rather than $(\states, \set{\initstate}, \analph, \transitions, \finalstates)$ and $\transitions(\astate, \aletter)=\set{\astatep}$.

\subsection*{Synchronized Products}

Let $\anults_1$ and $\anults_2$ be LTS with $\anults_i=(\states_i, \initstates_i, \Sigma, \transitions_i, \finalstates_i)$.
We define their synchronized product as the LTS
$\anults_1\times\anults_2=(\states_1\times\states_2,
\initstates_1\times\initstates_2, \Sigma, \transitions, \finalstates_1\times\finalstates_2)$, 
where
\begin{align*}
(\astate_1, \astate_2) \trans{a} (\astatep_1, \astatep_2) \text{ in } \anults_1\times \anults_2,\qquad \text{if}\quad \astate_1 \trans{a} \astatep_1 \text{ in } \anults_1\quad \text{and}  \quad \astate_2 \trans{a} \astatep_2 \text{ in } \anults_2
        \ .
\end{align*}

It is immediate from the definition that the language of the product is the intersection of the languages, $\langof{\anults_1\times\anults_2}=\langof{\anults_1}\cap\langof{\anults_2}$. 
If $\anults_1$ and $\anults_2$ both are finitely branching, then so is their product.

\paragraph*{Compatibility}
We work with LTS $\anults=(\states, \initstates, \analph, \transitions, \finalstates)$ whose states form a quasi order $(\states, \leq)$ that is \emph{upward compatible} with the remaining components as follows.
We have $F=\upcls{F}$, the final states are upward closed wrt. $\leq$.
Moreover, $\leq$ is a simulation relation~\cite{M71}: for all pairs of related states $\astate_1\leq \astatep_1$ and for all letters $\aletter\in\Sigma$ we have that
\begin{align*}
\text{for all }\astate_2\in\transitions(\astate_1, \aletter)
    \text{ there is }\astatep_2\in\transitions(\astatep_1, \aletter)\text{ with }\astate_2\leq\astatep_2\ .
\end{align*}
We also make the quasi order explicit and call $\anults=(\states, \leq, \initstates, \analph, \transitions, \finalstates)$ an \emph{upward-compatible} LTS (ULTS).
The synchronized product of ULTS is again a ULTS (with the product order). 

ULTS can be determinized, in the case of $\anults$ from above this yields
\begin{align*}
    \detof{\anults}\quad=\quad(\downpowof{\states}, \subseteq, \downcls{I}, \analph, \detof{\transitions}, \detof{F})\ .
\end{align*}
The states are the downward-closed sets ordered by inclusion, the transition relation is defined by closing the result of the original transition relation downwards, $\detof{\transitions}(\adownset, \aletter)= \downcls{\transitions(\adownset, \aletter)}$ for all $\adownset\in \downpowof{\states}$ and $\aletter\in\analph$, and the set of final states consists of all downward-closed sets that contain a final state in the original ULTS, $\detof{F}=\setcond{\adownset\in \downpowof{\states}}{\adownset\cap \finalstates\neq\emptyset}$.
\begin{lem}\label{Lemma:Determinization}
Let $\anults$ be an ULTS. Then $\detof{\anults}$ is a deterministic ULTS with $\langof{\detof{\anults}}=\langof{\anults}$.
\end{lem}
We write $\udlts$ for the class of deterministic ULTS.
%
\paragraph*{Well Quasi Orders}
We call a quasi order $(\states, \leq)$ a well quasi order~(\wqo), if for every infinite sequence $\sequence{\astate_{i}}{i\in\nat}$ in $\states$ there are indices $i<j$ with $\astate_{i}\leq\astate_{j}$.
It is folklore that $(\states, \leq)$ is \wqo iff it neither admits 
an infinite descending sequence (i.e.~it is well-founded) nor an infinite antichain.  
The set of downward-closed subsets of a \wqo, ordered by inclusion, is not necessarily a \wqo.
This only holds for so-called $\omega^2$-\wqos~\cite{AMarconeBQO}.
We forego the definition and work with the following characterization of $\omega^2$-\wqos:
\begin{lem}[\cite{JANCAR1999}] \label{lem:omega2}
    $(\states, \leq)$ is an $\omega^2$-\wqo iff $\big(\downpowof{\states}, {\subseteq}\big)$ is a \wqo.
\end{lem}
However, the \wqo property is preserved if we restrict our focus to the downward closures of finite sets.
The set $\findownpowof{\states}$, ordered by inclusion, is a \wqo whenever $(\states, \leq)$ is:
\begin{lem}
\label{c:pfin}
    $\big(\findownpowof{\states}, {\subseteq}\big)$ is a \wqo iff $(\states, \leq)$ is a \wqo.
\end{lem}
As a matter of fact, \cite{JANCAR1999} considers the reverse inclusion order on upward-closed sets, which is clearly isomorphic to
the inclusion order on downward-closed sets.
\paragraph*{Well-Structuredness}
An \emph{upward-compatible well-structured transition system} (WSTS) is an ULTS $\anults$ whose states $(\states, \leq)$ form a WQO.
An $\upomegasqwsts$ is an ULTS whose states form an $\omega^2$-\wqo.
The synchronized product of WSTS is again a WSTS. 
We are interested in $\langof{\upwsts}$, the class of all languages accepted by WSTS. 
We also study the subclasses $\langof{\upfinbranchwsts}$, $\langof{\upomegasqwsts}$, and $\langof{\updwsts}$ of languages accepted by finitely branching \wsts, $\omega^2$-\wsts, and deterministic \wsts, respectively. %

In our definition of WSTS, each transition is labeled by a letter $\aletter\in\analph$. 
WSTS have also been defined to admit $\varepsilon$-transitions that can be taken without reading a letter \cite{WSTSLanguages07}.
The definition of upward compatibility is then generalized as follows: for all transitions $\astatep_0\trans{\aletter}\astatep_1$ with $\aletter\in\analph\cup\set{\varepsilon}$ and all $\astatep_0\leq\astatepp_0$, there is a sequence $\astatepp_0\trans{b_1}\ldots \trans{b_k}\astatepp_{k}$ so that $\astatep_1\leq\astatepp_{k}$ and $a=b_1\ldots b_{k}$.
An $a$-labeled transition from a smaller state is simulated by a sequence of transitions from the larger state, where exactly one transition is labeled by $a$ and the remaining transitions are labeled by $\varepsilon$.

As we show in the next lemma, $\varepsilon$-transitions do not make the WSTS model more expressive when it comes to languages.
This means our separability result will also hold for WSTS with $\varepsilon$-transitions.
We decided to exclude $\varepsilon$-transitions from our WSTS definition as they make studying properties like finite branching harder.
\begin{lem}\label{Lemma:EpsilonTransitions}
    For every WSTS with $\varepsilon$-transitions $\anults$, there is a WSTS without $\varepsilon$-transitions~$\anultsp$ so that $\langof{\anults}=\langof{\anultsp}$.
\end{lem}
\noindent To obtain $\anultsp$ from $\anults$, it suffices to modify the transitions: we include $\astate\trans{a}\astatep$ in $\anultsp$ if a sequence of transitions labeled by $a$ reaches $\astatep$ from $\astate$ in $\anults$. 
The $\varepsilon$-transitions are deleted.
To show upward compatibility, we consider a new transition that should be mimicked and use an induction on the length of the sequence of old transitions that leads to this new transition with the aforementioned construction.

Finally, we observe that we can focus on WSTS with a countable number of states.
\begin{lem}\label{Lemma:Countability}
    For every $\alang\in\langof{\upwsts}$ there is a WSTS $\anults$ with a countable number of states so that $\alang=\langof{\anults}$.
\end{lem}
\noindent The lemma needs two arguments: the language consists of a countable number of words,
and we can assume the transition relation to yield downward-closed sets.

\subsection*{Examples of WSTS}

Various well-studied models of computation happen to be either \UWSTS or the downward-compatible relative that we introduce later on. 
The list contains Petri nets/VAS and their extensions (e.g.~with reset arcs or transfer arcs)~\cite{DufourdFinkelReset}, lossy counter machines~\cite{BouajjaniM99}, 
string rewriting systems based on context-free grammars~\cite{WSTSEverywhere01}, 
lossy channel systems~\cite{BrandZ83,AJ93}, and more.
In the first two models listed above, the states are ordered by the multiset embedding, while in the remaining two 
the states are ordered by Higman's subsequence ordering~\cite{Higman}.
The natural examples of \UWSTS, including all models listed above, are $\omega^2$-\UWSTSes and, when considered without $\varepsilon$-transitions, finitely branching.


\section{Regular Separability}\label{Section:RegSepPrinciple}
Two languages $\alang_1, \alang_2\subseteq\analph^*$ are \emph{separable by a regular language}, denoted by $\alang_1\regsep\alang_2$, if there is a regular language $\areg\subseteq\analph^*$ with $\alang_1\subseteq \areg$ and $\areg\cap\alang_2=\emptyset$.
Our main result is that disjoint \wsts\ languages are always separable in this sense.

\begin{thm}\label{Theorem:RegSep}
    For $\alang_1, \alang_2\in\langof{\upwsts}$, we have $\alang_1\regsep\alang_2$ if and only if $\alang_1\cap\alang_2=\emptyset$.
\end{thm}

\subsection{Proof Principle}

Theorem~\ref{Theorem:RegSep} follows from a surprising link between 
the regular separators we are interested in and inductive invariants as known from verification: every inductive invariant (for the product of the systems) can be turned into a regular separator. 
The result, stated as Theorem~\ref{thm:core} below, needs the assumption that one of the systems is deterministic and the invariant can be represented in a finite way.  
We give the definitions related to invariants and then make the result formal. 
\paragraph*{Inductive Invariants}
Inductive invariants are a standard tool in the safety verification of programs~\cite{MP1995}.
An inductive invariant (of a program for a safety property) is a set of program states that includes the initial ones, is closed under the transition relation, and is disjoint from the set of undesirable states.
The following definition lifts the notion to \WSTSes, actually to the more general \ULTSes, where it is natural to require inductive invariants to be downward closed.

\begin{defi}
    An \emph{inductive invariant} for a \ULTS $U$ with states $S$ is a downward-closed set $X\subseteq S$ with the following properties:
    \begin{align}
        \label{eq:iiI}
        &I \subseteq X
        \ ,
        \\
        \label{eq:iiF}
        &F \cap X = \emptyset
        \ ,
        \\
        \label{eq:iiSucc}
        &\succe{\W}{X}{a} \subseteq X \ \text{for all } a\in\Sigma
        \ .
    \end{align}
    An inductive invariant $X$ is \emph{finitely represented}, if $X = \downclosure{Q}$ for a finite set $Q \subseteq S$.
\end{defi}

Consider two \ULTSes with disjoint languages.
Theorem~\ref{thm:core} tells us that any finitely-represented inductive invariant for the product can be turned into a regular separator.
Interestingly, the theorem does not need the WQO assumption of WSTS but holds for general \ults.
\begin{thm}
\label{thm:core}
    Let $\anults$ and $\anultsp$ be \ULTSes, one of them deterministic.
    If $\anults \times \anultsp$ admits a finitely-represented inductive invariant, then $\langof{\anults}\regsep\langof{\anultsp}$ holds. 
    Moreover, the states of the separating automaton are the states of $\anults \times \anultsp$ that represent the invariant. 
\end{thm}

There is a tension between the two premises, the finite representation of the invariant and determinism, which makes  Theorem~\ref{thm:core} difficult to apply   in order to prove Theorem~\ref{Theorem:RegSep}.
If we have two \wsts\ rather than the more general \ULTSes, and one of them is guaranteed to be deterministic, then it is easy to represent an invariant in a finite way using ideal decompositions.
We explore this special case in Section~\ref{Section:SimpleRegSep}. 
However, if both \wsts\ are non-deterministic, we first have to determinize one of them. 
As Lemma~\ref{Lemma:Determinization} shows, this still yields an \ULTS\ 
(and so Theorem~\ref{thm:core} still applies).
Unfortunately, as we prove in Section~\ref{Section:NonDet}, the determinization will ruin the WQO property for some \wsts.
Without a WQO on the states, it is harder to represent invariants in a finite way, and we need Sections~\ref{Section:CTS} and \ref{Section:FullRegSep} to achieve this.
In the remainder of this section, we prove Theorem~\ref{thm:core}.
\subsection{Proving the Proof Principle}\label{Subsection:ProofOfProofPrinciple}
We construct a separator starting from the premises of Theorem~\ref{thm:core}, and prove its correctness.
Let $\anults = (\states_{\anults}, \leq, \initstates_{\anults},\analph, \transitions_{\anults}, \finalstates_{\anults})$ be an arbitrary \ULTS and $\anultsp = (\states_{\anultsp}, \preceq, \initstates_{\anultsp},\analph, \transitions_{\anultsp}, \finalstates_{\anultsp})$ be a determinstic one.
Assume $\langof{\anults}\cap\langof{\anultsp}=\emptyset$. 
Let their synchronized product be
\begin{align*}
    \prodW = \anults \times \anultsp = (\prodstates, \prodleq, \prodinitstates, \analph, \prodtransitions, \prodleq, \prodfinalstates)\ .
\end{align*}
By disjointness, we have $\lang{\prodW} = \emptyset$.
Let $\sepstates \subseteq \prodstates$ be the finite set of states so that $\downclosure{\sepstates}$ is an inductive invariant of $\prodW$.
\paragraph*{Constructing the Separator}
We define a finite automaton $\annfa$ whose set of states is the set $\sepstates$ representing the invariant. 
As required for a separator, the language of $\annfa$ will contain~$\langof{\anults}$ and be disjoint from $\langof{\anultsp}$.
The idea is to over-approximate the states of $\prodW$ by the elements available in $\sepstates$.
The fact that $\reachall{\prodW} \subseteq \downclosure{\sepstates}$ guarantees that every state $(s, s')\in \prodstates$
has such a representation.
Since we seek to approximate the language of $\anults$, the final states only refer to the $\anults$-component.
The transitions are approximated existentially.

We define the \emph{separating automaton induced by $Q$} to be $\annfa = (\sepstates, \sepstates_I, \analph, \to, \sepstates_{F})$. 
We have $\sepstates_I = \setof{(s, s') \in Q}{(i, i') \prodleq (s, s') \text{ for some } (i, i') \in \prodinitstates}$, a state is initial if it dominates an initial state of $\prodW$. 
As final states, we take the pairs whose $\anults$-component is final,
$\sepstates_F = \setof{(s, s') \in Q}{ s \in \finalstates_U}$. 
Finally, the transition relation in $\A$ is an over-approximation of the transition relation in $\prodW$:
\begin{align*}
    (s, s') \trans{a} (r, r') \text{ in } \A \quad \text{ if } (s, s') \trans{a} (t, t') \text{ in } \prodW \text{ for some } (t, t') \prodleq (r, r')\ .
\end{align*}
Figure~\ref{Figure:SepTransRel} illustrates the construction.

\begin{figure}[b]
    \begin{align*}
    \xymatrix
    {
        && (r,r') \in Q \\
        Q \ni (s,s') \ar@/^1pc/@{.>}[rru]^a_{\text{ in } \annfa} \ar[rr]^a_{\text{ in } \prodW} &&
        (t,t') \in \prodS \ar@{}[u]|{\rot{\prodleq} \qquad}
    }
    \end{align*}
    \caption{Transition relation of $\annfa$.}
    \label{Figure:SepTransRel}
\end{figure}

\subsection*{Proving Separation}
We need to prove $\langof{\anults} \subseteq \langof{\annfa}$ and \mbox{$\langof{\annfa}\cap \langof{\anultsp} = \emptyset$}, and begin with the former.
As $\anultsp$ is deterministic, $\prodW$ contains all computations of $\anults$.
Due to the upward compatibility, $\annfa$ over-approximates the computations in $\prodW$.
Combining these two insights, which are summarized in the next lemma, yields the result.

\begin{lem}
\label{lem:sim}
    (1) For every $s \in \reach{\anults}{w}$ there is some \mbox{$(s, s') \in \reach{\prodW}{w}$}.
    (2) For every $(s, s') \in \reach{\prodW}{w}$ there is some $(r, r') \in \reach{\annfa}{w}$ with $(s, s') \prodleq (r, r')$.
\end{lem}
\begin{prop}
\label{prop:SepW}
    $\langof{U} \subseteq \langof{\A}$\ .
\end{prop}

For the disjointness of $\langof{\A}$ and $\langof{\anultsp}$, the key observation is this.
Due to determinism, $\anultsp$ simulates the computations of $\annfa$ in the following sense:
if upon reading a word $\annfa$ reaches a state $(s, s')$, then the unique computation of $\anultsp$ will reach a state dominated by $s'$.

\begin{lem}
\label{lem:invsimWprime}
    For every $w \in \Sigma^*$ and every $(s, s') \in \reach{\A}{w}$, we have $\reach{\anultsp}{w} \preceq' s'$.
\end{lem}
\noindent With this lemma at hand,  we can show disjointness.
Towards a contradiction, suppose there is a word $w\in \lang{\A} \cap \lang{\anultsp}$.
As $w\in \lang{\annfa}$, there is a state $(s, s') \in \reach{\annfa}{w}$ with $s \in \finalstates_{\anults}$.
As $w\in \lang{\anultsp}$, the unique state $\reach{\anultsp}{w}$ belongs to $\finalstates_{\anultsp}$.
With the previous lemma and the fact that $F_{\anultsp}$ is upward closed, we conclude $s'\in \finalstates_{\anultsp}$.
Together, $(s, s')\in \prodfinalstates$, which contradicts the fact that $\downclosure{\sepstates}$ is an inductive invariant, Property~\eqref{eq:iiF}.

\begin{prop}
\label{prop:SepWprime}
    $\langof{\A} \cap \langof{\anultsp} = \emptyset$\ .
\end{prop}
\noindent
Together, Proposition~\ref{prop:SepW} and~\ref{prop:SepWprime} show Theorem~\ref{thm:core}.
\subsection{Discussion}\label{Subsection:Problems}
Thanks to  Theorem~\ref{thm:core}, the proof of our separability theorem takes the form of an invariant construction: we look for a finitely-represented inductive invariant and use Theorem~\ref{thm:core} to convert it into a separator.
Fortunately, there is a guarantee that an inductive invariant always exists, though it may not be representable in a finite way.
The reason is that inductive invariants are complete for proving emptiness, like inductive invariants for programs are complete for proving safety~\cite{Cook1978}. 

\begin{lem}
\label{Lemma:Invariants}
    Let $U$ be a \ULTS.
    Then $\langof{U}=\emptyset$ iff there is an inductive invariant for~$U$.
\end{lem}
\noindent

To see the if-direction, let $X$ be an inductive invariant for $U$.
By~\eqref{eq:iiI} and~\eqref{eq:iiSucc}, the invariant has to contain the whole reachability set.
By~\eqref{eq:iiF} and~\eqref{eq:iiSucc}, it has to be disjoint from the predecessors of the final states:
\begin{align*}
    &\reachall{\W} \subseteq X
    &
    &\revreachall{\W} \cap X = \emptyset
    \ .
\end{align*}
Combined, this shows language emptiness.

For the only if-direction, observe that $X = \downclosure{\reachall{\W}}$ is an inductive invariant if the language is empty. 
Actually, it is the least one wrt.~inclusion.
There is also a greatest inductive invariant, namely the complement of $\revreachall{\W}$.
Other invariants may have the advantage of being easier to represent.
Our goal is to find one that can be represented as the downward closure of a finite set. 

\section{A Deterministic Model for WSTS}\label{Section:CTS}
We just discussed that, in order to prove our separability result with the help of Theorem~\ref{thm:core}, we need to work with a deterministic model.
We propose here converging transition systems~(CTS) as a new class of ULTS that is general enough to capture determinized WSTS and retains enough structure to establish the existence of finitely-represented inductive invariants. 
The states of a CTS form a powerset lattice in which every infinite sequence contains a subsequence that converges in a natural sense.
CTS are inspired by Noetherian transition systems~\cite{GL07,GL10}, but are formulated in a lattice-theoretic rather than in a topological way.

We summarize the overall argumentation for our separability result.
We only define CTS as a deterministic model, and therefore have to determinize both of the given \wsts.
CTS are closed under products. 
Since the given \wsts\ languages are disjoint by assumption, the product has an empty language and, by Lemma~\ref{Lemma:Invariants}, an inductive invariant. 
It remains to turn this invariant of the product CTS into an invariant that can be represented in a finite way. 
This will be the topic of Section~\ref{Section:FullRegSep}, here we only give the idea. 
We add to the invariant the limits of all converging sequences (that live inside the invariant).
Since the CTS transitions are compatible with limits, the resulting set of states is again an inductive invariant.
An application of Zorn's lemma shows that the maximal elements in this invariant form the finite representation that was needed to conclude the proof.

The reader may wonder whether CTS are actually needed for the above argumentation.
The answer is no, one could give the proof directly on ULTS. 
These ULTS, however, would be the result of determinizing the given WSTS and forming the product.
This means they would  operate on products of powerset lattices that satisfy a particular convergence condition.
CTS allow us to abstract away the product and the powerset structure and highlight the key arguments in the closure of a given invariant under forming limits. 
Still, to have a more concrete account of what follows, it will help to remember this instantiation.

%
%
%

\paragraph*{Converging Lattices}
Recall that determinized WSTS have as state space~$(\downpowof{\states}, \subseteq)$ with $(\states,\leq)$ a WQO.
%
%
In a WQO, every infinite sequence admits an increasing subsequence.
It is well-known that this may not hold for $(\downpowof{\states}, \subseteq)$~\cite{Rado54}.
However, a natural relaxation holds: every infinite sequence~$\sequence{\asetn{i}}{i\in\nat}$
admits an infinite subsequence $\sequence{\asetn{\varphi(i)}}{i\in\nat}$,
where every element that is present in one set is present in almost every set.
A similar property, defined for complete lattices, is called convergence in the literature \cite{ConvergenceLattices}.
Our definition differs from the citation in two ways.
We restrict ourselves to sequences (as opposed to nets), and we require convergence to the join (as opposed to $\lim\sup=\lim\inf$).
This suffices for us.
%
%
%

\begin{defi}\label{Definition:ConvergingLattice}
    A \emph{converging lattice} $(\states, \leq)$
    is a completely distributive lattice,
    where every sequence $\sequence{\astate_i}{i\in\nat}$
    has a converging subsequence $\sequence{\astate_{\varphi(i)}}{i\in \nat}$.
    A \emph{converging sequence} $\sequence{\astatep_i}{i\in \nat}$ is an infinite
    sequence with
    \[
        \biglatjoin_{i\in \nat}\biglatmeet_{j\geq i}\astatep_j=\biglatjoin_{i\in\nat}\astatep_i\ .
    \]
\end{defi}
\noindent The equality formalizes our explanation.
In the context of sets, where join and meet are
union and intersection, respectively, the right-hand side of the equation contains all elements that
appear in some set in the sequence.
The left side iterates over every finite initial segment,
and includes every element that appears in all sets outside of this
segment.
This means every element that is missing in only finitely many sets will
eventually be included.

Converging lattices not only represent downward-closed subsets of WQOs, 
they are even a precise representation in the following sense: whenever we have a converging lattice of downward-closed sets, then the underlying structure has to be a WQO. 
\begin{lem}\label{Lemma:ConvergenceInWqo}
$(\downpowof{\states}, \subseteq)$ is
    a converging lattice if and only if $(\states, \leq)$ is a WQO.
\end{lem}
\noindent The backward direction is by \cite[Proof of Theorem 3]{Rado54}, see also \cite[Lemma 2.12]{Pequignot15}. 
The forward direction is by the following fact~\cite{Rado54}, also \cite[Fact III.3]{LS15}: $(\downpowof{\states}, \subseteq)$ is well-founded if and only if the order is a WQO.
Convergence clearly implies well-foundedness. 

The space of converging sequences is closed under the application of join-preserving functions as formulated next.
While we would expect this result to be known, we have not found a reference.
The lemma is central to our argument, therefore we give the proof.
\begin{lem}\label{Lemma:ConvergenceContinuous}
Let $(\states, \leq)$ be a lattice, $\sequence{\astate_i}{i\in\nat}$ a converging sequence in $\states$,
and $f:\states\rightarrow\states$ a join-preserving function.
Then also $\sequence{f(\astate_i)}{i\in\nat}$ is converging.
\end{lem}
\begin{proof}
Due to convergence of the given sequence, we have $\biglatjoin_{i\in\nat}\biglatmeet_{j\geq i}\astate_j =
        \biglatjoin_{i\in\nat}\astate_i$.
This equality yields
$f(\biglatjoin_{i\in\nat}\biglatmeet_{j\geq i}\astate_j) =
        f(\biglatjoin_{i\in\nat}\astate_i)$.
By the join preservation of $f$, we get
\begin{align*}
        \biglatjoin_{i\in\nat}f(\biglatmeet_{j\geq i}\astate_j) =
        \biglatjoin_{i\in\nat}f(\astate_i)\ .
\end{align*}

Function $f$ may not be meet preserving. 
However, as it is join preserving, it is also monotonic.
We thus have $f(\biglatmeet S')\leq f(\astateppp)$ for all $\astateppp\in S'$.
This means $f(\biglatmeet S')\leq\biglatmeet_{s\in S'}f(\astateppp)$.
We apply this inequality to the previous equality:
\begin{align*}
        \biglatjoin_{i\in\nat}f(\astate_i) =
        \biglatjoin_{i\in\nat}f(\biglatmeet_{j\geq i}\astate_j)\leq
        \biglatjoin_{i\in\nat}\biglatmeet_{j\geq i}f(\astate_j)\leq
        \biglatjoin_{i\in\nat}f(\astate_i)\ .
\end{align*}
This is $\biglatjoin_{i\in\nat}\biglatmeet_{j\geq i}f(\astate_j)=\biglatjoin_{i\in\nat}f(\astate_i)$, as desired.
\end{proof}

\paragraph*{CTS}
We explain the considerations that lead us to the definition of CTS given below.
In the light of Lemma~\ref{Lemma:ConvergenceInWqo}, the states of a CTS should form a converging lattice.
This, however, was not enough to guarantee the existence of finitely-represented inductive invariants.
Invariants should be closed under taking transitions.
To understand which sets satisfy this requirement, we had to restrict the transition relation.
We only define CTS as a deterministic model.
Then the transitions form a function $\delta(- , \aletter)$ for every letter~$\aletter\in\analph$.
Upward compatibility of these functions is not very informative.
Consider determinized WSTS: upward compatibility would give us $\delta(S_0\cup S_1, \aletter)\supseteq\delta(S_0, \aletter)$, while we expect
$\delta(S_0\cup S_1, \aletter)=\delta(S_0, \aletter)\cup \delta(S_1, \aletter)$.
In lattice-theoretic terms, we expect the transition functions $\delta(-, \aletter)$ to be join preserving.
A benefit of this requirement is of course that it makes Lemma~\ref{Lemma:ConvergenceContinuous} available.
An invariant should also be disjoint from the final states, so we had to control this set as well.
When we determinize a WSTS, a set $\adownset\in\downpowof{\states}$ becomes final as soon as it contains a single final state.
Given the definition of convergence, we relax this to containing a finite set of final states.
\begin{defi}\label{Definition:CTS}
    A \emph{converging transition system} (CTS) is an ULTS
    $\anults=(\states, \leq, {\initstate}, \analph,
    \transitions, \finalstates)$ that is deterministic, where $(\states, \leq)$
    is a converging lattice, the functions $\transitions(-, \aletter)$ are join preserving for all $\aletter\in\analph$,
    and the final states satisfy
    \begin{align*}
        \text{finite acceptance: }&\text{ for every }\biglatjoin K \in \finalstates\text{ there is a finite set }
        N \subseteq K
        \text{ with }\biglatjoin N\in\finalstates\ .
    \end{align*}
\end{defi}

With this definition, the determinization of a WSTS yields a CTS.
Somewhat surprisingly, CTS do not add expressiveness but their languages are already accepted by (non-deterministic) WSTS.
The construction is via join prime elements and can be found in the full version~\cite{keskinmeyer2023sep}.
Together, the CTS languages are precisely the WSTS languages, and one may see Definition~\ref{Definition:CTS} as a reformulation of the WSTS model.
\begin{prop}\label{Proposition:Correspondence}
If $\anults$ is a WSTS, then $\detof{\anults}$ is a CTS with $\langof{\anults}=\langof{\detof{\anults}}$. 
For every CTS~$\anultsp$, there is a WSTS~$\anults$ with $\langof{\anultsp}=\langof{\anults}$. 
Together, $\langof{\wsts}=\langof{\cts}$.
\end{prop}

The correspondence allows us to import the countability assumption from Lemma~\ref{Lemma:Countability}.
Indeed, if the WQO $(\states, \leq)$ is countable, then there is only a countable number of downward-closed sets in $(\downpowof{\states}, \subseteq)$.
This is by a standard argument for WSTS: each downward-closed set can be characterized by its complement, the complement is upward closed, and is therefore characterized by its finite set of minimal elements.
\begin{lem}\label{Lemma:CountableCTS}
For every $\alang\in\langof{\cts}$, there is a CTS $\anults$ over a countable number of states so that $\alang=\langof{\anults}$
\end{lem}

We will also need that CTS are closed under synchronized products.
\begin{lem}\label{Lemma:ProductCTS}
If $\anults$ and $\anultsp$ are CTS, so is $\anults\times\anultsp$.
\end{lem}

We summarize the findings so far.
Given disjoint WSTS languages $\langof{\anultsp_1}\cap\langof{\anultsp_2}=\emptyset$, the goal is to show regular separability $\langof{\anultsp_1}\regsep\langof{\anultsp_2}$.
We first determinize both WSTS.
By Proposition~\ref{Proposition:Correspondence}, $\detof{\anultsp_1}$ and $\detof{\anultsp_2}$ are language-equivalent CTS.
We use Lemma~\ref{Lemma:CountableCTS} to obtain countable CTS $\anults_1$ and $\anults_2$ that still accept the orginal languages.
To show regular separability, we now intend to invoke Theorem~\ref{thm:core} on $\anults_1$ and~$\anults_2$.
Since CTS are deterministic by definition, it remains to show that $\anults_1\times\anults_2$ has a finitely-represented inductive invariant.
With Lemma~\ref{Lemma:ProductCTS}, $\anults_1\times\anults_2$ is another CTS $\anults$.
Moreover, the product corresponds to language intersection, and so $\langof{\anults}=\emptyset$.
By Lemma~\ref{Lemma:Invariants}, we know that~$\anults$ has an inductive invariant.
We now show how to turn this invariant into a finitely-represented one.

\section{Finitely Represented Invariants in CTS}\label{Section:FullRegSep}

We show the following surprising property for countable CTS: every inductive invariant $\inductiveinv$ can be generalized to an inductive invariant $\closureof{\inductiveinv}$ that is finitely represented.
The closure operator is defined by adding to $\inductiveinv$ the joins of all converging sequences:
\begin{align*}
\closureof{\inductiveinv}\quad=\quad\setcond{\ \biglatjoin_{i\in\nat}\astate_{i}}{\sequence{\astate_{i}}{i\in\nat}\text{ a converging sequence in }\inductiveinv\ }\ .
\end{align*}
\begin{prop}\label{Proof:FinitelyRepresentedInvariants}
Let $\anults$ be a countable CTS and $\inductiveinv$ an inductive invariant of $\anults$. Then also $\closureof{\inductiveinv}$ is an inductive invariant of $\anults$ and it is finitely represented.
\end{prop}
The proposition concludes the proof of Theorem~\ref{Theorem:RegSep}.
We simply invoke it on the inductive invariant that exists by Lemma~\ref{Lemma:Invariants} as discussed above.
The rest of the section is devoted to the proof.
We fix a countable CTS $\anults=(\states, \leq, \initstate, \analph, \transitions, \finalstates)$
and an inductive invariant $\inductiveinv\subseteq\states$.

As Lemma~\ref{Lemma:IdempotenceAndDC} states, the closure is expansive and idempotent.
This means further applications do not add new limits. 
Here, we need the fact that we have a completely distributive lattice.
Moreover, the closure yields a downward-closed set.
The closure is also trivially monotonic, and hence an upper closure operator indeed~\cite[Section 11.7]{C21}, but we will not need monotonicity.
The proof of Lemma~\ref{Lemma:IdempotenceAndDC} is given in~\cite{keskinmeyer2023sep}.
\begin{lem}\label{Lemma:IdempotenceAndDC}
    $\inductiveinv\subseteq\completedinv{\inductiveinv}=\completedinv{\completedinv{\inductiveinv}}
    =\downcls{\completedinv{\inductiveinv}}$.
\end{lem}

Towards showing Proposition~\ref{Proof:FinitelyRepresentedInvariants}, we first argue for invariance.

\begin{lem}
$\closureof{\inductiveinv}$ is an inductive invariant.
\end{lem}

\begin{proof}
    To prove that $\closureof{\inductiveinv}$ is an inductive invariant, we must show two properties
    for the joins $\biglatjoin_{i\in\nat}{\astate_i}=\astate$ of converging sequences $\sequence{\astate_{i}}{i\in\nat}$ in~$\inductiveinv$ that we added.
    First, we must show that we do not leave $\completedinv{\inductiveinv}$ when taking transitions,
    $\apply{\astate}{\aletter}\in
    \completedinv{\inductiveinv}$
    for all $\aletter\in\analph$.
    Second, we must show that the join is not a final state.
    We begin with the latter.
    Towards a contradiction, suppose $\astate\in\finalstates$.
    Convergence yields $\biglatjoin_{i\in\nat}\biglatmeet_{j\geq i}\astate_j\in\finalstates$.
    By the finite acceptance property of CTS, there must be a finite set $K\subseteq\nat$ with
    $k=\max K$ so that
    \[
        \biglatjoin_{i\in K}\biglatmeet_{j \geq i}\astate_j\ =\
        \biglatmeet_{j\geq k} \astate_j
        \in\finalstates\ .
    \]
    Since $\biglatmeet_{j\geq k}\astate_j \leq \astate_k$ and $\finalstates$ is upward closed, we obtain $\astate_k\in\finalstates$.
    This is a contradiction: $\astate_k$ belongs to the inductive invariant
    $\inductiveinv$ and the invariant does not intersect the final states.

    To show $\apply{\astate}{\aletter}\in\completedinv{\inductiveinv}$, we use
    \begin{align*}
        \transitions(\astate, \aletter)\ =\
        \transitions(\biglatjoin_{i\in \nat}\astate_i, \aletter)\ =\
        \biglatjoin_{i\in\nat}\transitions(\astate_i, \aletter)\ \in\ \completedinv{\inductiveinv}\ .
    \end{align*}
    The first equality is by the definition of $\astate$, the following is by the fact that the transition function $\transitions(-, \aletter)$ is join preserving. 
    For membership in the closure, note that $\sequence{\astate_i}{i\in\nat}$ converges and the transition function is join preserving.
    By Lemma~\ref{Lemma:ConvergenceContinuous}, also  $\sequence{\transitions(\astate_i, \aletter)}{i\in\nat}$ converges.
    As~$\inductiveinv$ is an invariant and $\astate_i\in\inductiveinv$, we have $\transitions(\astate_i, \aletter)\in\inductiveinv$ for all $i\in\nat$.
    This means $\sequence{\transitions(\astate_i, \aletter)}{i\in\nat}$ is a converging sequence in the invariant, and so its join is added to the closure.
\end{proof}

It only remains to show that $\completedinv{\inductiveinv}$ is finitely represented.
In the following, we refer to a set $\acover\subseteq\completedinv{\inductiveinv}$ with $\downcls{\acover}=\completedinv{\inductiveinv}$ as a \emph{cover} of $\completedinv{\inductiveinv}$.

\begin{prop}\label{Proposition:FiniteCover}
    There is a finite cover of $\completedinv{\inductiveinv}$.
\end{prop}

We break down the proof of Proposition~\ref{Proposition:FiniteCover}
into two steps.
First, we show that $\completedinv{\inductiveinv}$
can be covered by an antichain.
Then, we show that infinite antichain covers do not exist.
This implies that there must be a finite antichain cover.
The proofs reason over \emph{closed} sets, sets that contain the limits of their converging sequences.
We rely on the fact that closed sets have at least one maximal element.
\begin{lem}\label{Lemma:Maximality}
    Consider
    $G\subseteq\states$ closed and non-empty. Then $\max G\neq\emptyset$.
\end{lem}
Moreover, closedness is preserved when subtracting a downward-closed set.
\begin{lem}\label{Lemma:DownDiffClosed}
    Consider $G, H\subseteq\states$
    where $G$ is closed.
    Then $G\setminus\downcls{H}$ is closed.
\end{lem}
We postpone the proofs of these lemmas until after the proof of
Proposition~\ref{Proposition:FiniteCover}.

\begin{lem}
    There is an antichain cover of $\completedinv{\inductiveinv}$.
\end{lem}

\begin{proof}
    We claim that the maximal elements $\max\completedinv{\inductiveinv}$
    form an antichain cover of $\completedinv{\inductiveinv}$.
    It is clear that $\max\completedinv{\inductiveinv}$ is an antichain.
    Since $\completedinv{\inductiveinv}$ is downward closed by Lemma~\ref{Lemma:IdempotenceAndDC},
    we also have $\downcls{(\max\completedinv{\inductiveinv})}
    \subseteq\completedinv{\inductiveinv}$.
    To see that $\max\completedinv{\inductiveinv}$ is
    a cover, let $G=\completedinv{\inductiveinv}\setminus
    \downcls{(\max\completedinv{\inductiveinv})}$ and suppose $G\neq\emptyset$.
    Lemma~\ref{Lemma:DownDiffClosed} tells us that $G$ is closed.
    By Lemma~\ref{Lemma:Maximality}, we get $\max G\neq\emptyset$.
    Consider $\astate\in \max G$.
    By the definition of $G$, we have $\astate\not\in\max\completedinv{\inductiveinv}$.
    Then, however, there must be $\astatep\in\completedinv{\inductiveinv}$
    with $\astate\leq\astatep$ and $\astate\neq\astatep$.
    If $\astatep\in\downcls{(\max\completedinv{\inductiveinv})}$,
    then $\astate\in\downcls{(\max\completedinv{\inductiveinv})}$
    as well, which is a contradiction to $\astate\in G$.
    If conversely
    $\astatep\in
    \completedinv{\inductiveinv}\setminus\downcls{(\max\completedinv{\inductiveinv})}
    =G$,
    then we have a contradiction to $p\in\max G$.
\end{proof}

Now we prove the second part of Proposition~\ref{Proposition:FiniteCover}.
%
%
\begin{lem}
    There is no infinite antichain cover of $\completedinv{\inductiveinv}$.
\end{lem}

\begin{proof}
    Suppose there is an infinite
    antichain cover $\acover\subseteq\completedinv{\inductiveinv}$.
    Then, there is an infinite sequence
    $\sequence{\astate_i}{i\in \nat}$ in $\acover$
    in which no entry repeats.
    By Definition~\ref{Definition:ConvergingLattice},
    it has an infinite converging subsequence
    $\sequence{\astate_{\varphi(i)}}{i\in \nat}$.
    The closure operator adds
    $\biglatjoin_{i\in \nat}\astate_{\varphi(i)}$ to $\completedinv{\inductiveinv}$.
    Since $\acover$ is a cover of $\completedinv{\inductiveinv}$,
    there must be $\astatep\in\acover$ with
    $\biglatjoin_{i\in \nat}\astate_{\varphi(i)}\leq\astatep$.
    Hence, $\astate_{\varphi(i)}\leq\astatep$ for all~$i\in\nat$.
    Either $\astate_{\varphi(0)} = \astatep$, in which case $\astate_{\varphi(1)} < \astatep = \astate_{\varphi(0)}$, or $\astate_{\varphi(0)} < \astatep$.
    In both cases, we have found two comparable distinct elements of $\acover$.
    This contradicts the antichain property.
\end{proof}
We conclude by showing Lemma~\ref{Lemma:Maximality} and~\ref{Lemma:DownDiffClosed}.

\begin{proof}[Proof of Lemma~\ref{Lemma:Maximality}]
    Let $\emptyset\neq G\subseteq\states$ be closed.
    We prove $G$ chain complete, meaning for every chain $P\subseteq G$ the limit $\biglatjoin P$ is again in $G$.
    Then Zorn's lemma~\cite{SetTheory} applies and yields $\max G\neq \emptyset$.
    We have Zorn's lemma, because we agreed on the Axiom of Choice.
    Towards chain completeness, consider an increasing sequence $\sequence{\astate_i}{i\in\nat}$ in $G$.
    We prove that $\biglatjoin_{i\in\nat} \astate_i\in G$.
    For any $i\in\nat$, we have $\biglatmeet_{j\geq i}\astate_j=\astate_{i}$.
    Hence, replacing each meet with the smallest element shows convergence.
    Since $\sequence{\astate_i}{i\in\nat}$ converges
    and $G$ is closed, we have $\biglatjoin_{i\in\nat}\astate_i\in G$.

    Although we are in a countable setting, the argument for sequences does not yet cover all chains.
    The problem is that the counting processs may not respect the order.
    To see this, consider a chain $P\subseteq G$ of order type $\omega\cdot 2$.
    The chain is countable, but no counting process can respect the order.
    We now argue that still $\biglatjoin P\in G$.
    By \cite[Theorem~1]{ChainCompleteMarkowsky}, there is a (wrt.~inclusion)
    increasing sequence of subsets
    $\sequence{P_i}{i\in\nat}$ in~$\powof{P}$, where each $P_{i}$ is finite and $\bigcup_{i\in\nat}P_{i}=P$.
    Finite chains contain maximal elements, so let $p_i = \max P_{i}
    =\biglatjoin P_{i}$.
    Then
    \[
        \biglatjoin P\ =\ \biglatjoin \bigcup_{i\in\nat} P_i
        \ =\ \biglatjoin_{i\in\nat} \biglatjoin P_i\ =\ \biglatjoin_{i\in\nat} p_i\ .
    \]
    Since $\sequence{P_i}{i\in\nat}$ is an increasing sequence,
    $\sequence{p_i}{i\in\nat}$ is also an increasing
    sequence.
    As we have shown before, $\biglatjoin_{i\in\nat} p_i \in G$.
    This concludes the proof.
\end{proof}


\begin{proof}[Proof of Lemma~\ref{Lemma:DownDiffClosed}]
    Consider $G, H\subseteq\states$ with $G$ closed.
    We show that $G\setminus\downcls{H}$ is closed.
    Let $\sequence{\astate_i}{i\in\nat}$ be  a converging sequence in $G\setminus\downcls{H}$.
    Let $\astatep=\biglatjoin_{i\in\nat}\astate_i$ and
    suppose $\astatep\not\in G\setminus\downcls{H}$.
    Since $G$ is closed, $\astatep\in G$.
    Then necessarily~$\astatep\in\downcls{H}$.
    But by definition, $\astate_i\leq\astatep$
    for all $i\in\nat$.
    So $\astate_i\in\downcls{H}$ as well.
    This contradicts the fact that the sequence $\sequence{\astate_{i}}{i\in\nat}$ lives in $G\setminus\downcls{H}$.
\end{proof}

\section{A Closer Look at WSTS Determinization}\label{Section:NonDet}

We explore the expressive power of non-determinism in the context of WSTS.
The main finding is that there are \wsts\ languages that, as we prove, cannot be accepted by a deterministic WSTS.
This reinforces the proof of our regular separability result via CTS.
The non-determinizability result is non-trivial and requires a novel characterization of the deterministic WSTS languages.

As a second result, we identify important subclasses of WSTS that admit determinization.
This motivates us to reconsider (in the next section and for these subclasses) the construction of a finitely-represented inductive invariant, and hence a regular separator.
\subsection{Determinizability}
We begin with the positive results and show that finitely branching WSTS and $\omega^{2}$-WSTS are language-equivalent to deterministic WSTS.
Recall that the best studied WSTS models are in fact $\omega^{2}$-\WSTSes, meaning they all can be determinized.
The positive results also play a role in the proof of our non-determinizability result.
They inform us which properties a language outside $\langof{\updwsts}$ should not have.
\begin{thm}\label{Theorem:PositiveExpressivity}
    $\langof{\finbranchwsts}, \; \langof{\omegasqwsts}\subseteq \langof{\dwsts}$
\end{thm}

\begin{proof}
    Let $\anults=(\states, \leq, \initstates, \analph, \transitions, \finalstates)$ be a WSTS.
    If $\anults$ is an $\omega^{2}$-WSTS, the naive determinization shows $\langof{\anults}\in\langof{\dwsts}$.
    Indeed, $\detof{\anults}=(\downpowof{\states}, \subseteq, \detof{\initstates}, \analph, \detof{\transitions}, \detof{\finalstates})$ accepts the same language. Moverover, since $(\states, \leq)$ is an $\omega^{2}$-WQO, $(\downpowof{S}, \subseteq)$ is a WQO and so $\detof{\anults}$ a WSTS.

    Assume $\anults$ is finitely branching.
    To show $\langof{\anults}\in\langof{\dwsts}$, we again argue with~$\detof{\anults}$.
    We show that all states $X\in\reachall{\detof{\anults}}$ are downward closures of finite subsets of $\states$.
    This implies $\langof{\anults}\in\langof{\dwsts}$:
    we can restrict the states of $\detof{\anults}$ to $(\findownpowof{S}, \subseteq)$ and get~$\detof{\anults}_{\text{fin}}$ with the same language.
    By Lemma~\ref{c:pfin}, $(\findownpowof{S}, \subseteq)$ is a WQO, and so $\detof{\anults}_{\text{fin}}$ a WSTS.

    To prove that all reachable states are downward closures of finite subsets, we proceed by induction on the length of the computation.
    For the base case, we have the initial state $\downcls{\initstates}$ of $\detof{\anults}$.
    By the definition of finite branching, $\initstates$ is finite.
    For the inductive case, consider $X\in\reachall{\detof{\anults}}$ and~$a\in\Sigma$.
    The state reached from $X$ along $a$ in $\detof{\anults}$ is $\downcls{\succe{\anults}{X}{a}}$.
    By the induction hypothesis, there is a finite set $\acover\subseteq S$ with $X=\downcls{\acover}$.
    We claim that $\downcls{\succe{\W}{X}{a}}=\downcls{\acoverp}$ with $\acoverp=\succe{\W}{\acover}{a}$.
    As $\anults$ is finitely branching and $\acover$ is finite, we know that also $\acoverp$ is finite.
    We have $\acoverp=\succe{\W}{\acover}{a}\subseteq\succe{\W}{X}{a}$, because $\acover\subseteq X$.
    This implies $\downcls{\acoverp}\subseteq\downcls{\succe{\W}{X}{a}}$.
    For the reverse inclusion, let $\astatep\in\downcls{\succe{\W}{X}{a}}$.
    Then there is $\astate\in X$ and  $\astatep'\in S$ with $\astate\overset{a}{\to}\astatep'$ and $\astatep\preceq\astatep'$.
    Since $\downcls{\acover}=X$, there is also a state $\astate'\in\acover$ with $\astate\preceq\astate'$.
    By the simulation property of ULTS, we get $\astate'\overset{a}{\to}\astatep''$ for some $\astatep''\in \states$ with $\astatep'\preceq\astatep''$.
    By the definition of $\acoverp$, we have $\astatep''\in \acoverp$, and so $\astatep, \astatep'\in\downcls{\acoverp}$.
    This concludes the proof.
\end{proof}
\subsection{Non-Determinizability}
The positive results do not generalize to all WSTS, but we show that the \dwsts\ languages form a strict subclass of the \wsts\ languages.
To this end, we define a \wsts\ language $\witnesslang$, called  the \emph{witness language}, that we prove cannot be accepted by a \dwsts.
The proof relies on a novel characterization of the \dwsts\ languages that may be of independent interest.

\begin{thm}\label{Theorem:WSTSneqDWSTS}
$\langof{\dwsts}\neq\langof{\wsts}$.
\end{thm}

Towards the definition of $\witnesslang$, we just showed that finitely branching WSTS and $\omega^2$-WSTS can be determinized.
Moreover, it is known that $\omega^2$-WQOs are precisely the WQOs that do not embed the Rado WQO~\cite[Section 2]{JANCAR1999}.
This suggests we should accept the witness language $\witnesslang$ by an infinitely branching WSTS over the Rado WQO.
We begin with our characterization of the \dwsts\ languages, as it will provide additional guidance in the definition of the witness language.

\subsection{Characterization of the \dwsts\ Languages}
Our characterization is based on a classical concept in formal languages~\cite[Theorem 3.9]{HU79}.
The \emph{Nerode quasi order} $\mathord{\nerodeleqof{\alang}} \subseteq\analph^{*}\times\analph^{*}$ of a language $\alang\subseteq\analph^{*}$ is defined by $\aword\nerodeleqof{\alang}\awordp$, if
\begin{align*}
\text{for all }\awordpp\in\analph^*\text{ we have that $\aword\concat\awordpp\in\alang$ implies $\awordp\concat\awordpp\in\alang$\ .}
\end{align*}
The characterization says that the \dwsts\ languages are precisely the languages whose Nerode quasi order is a WQO.
This is not the folklore \cite[Chapter 11, Proposition 5.1]{HandbookFL1} saying that a language is regular if and only if the syntactic quasi order is a WQO.
%
\begin{lem}[Characterization of $\langof{\dwsts}$]\label{Lemma:WQONerode}
    $\alang\in\langof{\dwsts}$ iff \ $\mathord{\nerodeleqof{\alang}}$ is a WQO.
\end{lem}
\begin{proof}
For the only if-direction, consider $\alang = \langof{\anults}$ with $\anults=(\states, \leq, i, \analph, \transitions, \finalstates)$ a \dwsts.
We extend the order $\leq\ \subseteq\states\times \states$ on the states to an order $\leqof{\anults}\ \subseteq\analph^*\times\analph^*$ on words by setting $\aword\leqof{\anults}\awordp$, if $\transitions(i, \aword)=\astate$ and $\transitions(i, \awordp)=\astatep$ with $\astate\leq\astatep$.
Since $\anults$ is deterministic, the states $\astate$ and $\astatep$ are guaranteed to exist and to be unique.
It is easy to see that $\leqof{\anults}$ is a WQO.
We now show that $\leqof{\anults}$ is included in the Nerode quasi order, and so also $\nerodeleqof{\alang}$ is a WQO.
To this end, we consider $\aword\leqof{\anults}\awordp$ and $\awordpp\in\analph^{*}$ with
    $\aword\concat\awordpp\in\alang$, and show that also  $\awordp\concat\awordpp\in\alang$.
We have $\transitions(i, \aword\concat\awordpp)=\transitions(\astate_1, \awordpp)=\astate_2$ and $\transitions(i, \awordp\concat\awordpp)=\transitions(\astatep_1, \awordpp)=\astatep_2$ with $\astate_1=\transitions(i, \aword)$ and  $\astatep_1=\transitions(i, \awordp)$.
Since $\aword\leqof{\anults}\awordp$, we have $\astate_1\leq\astatep_1$.
With the simulation property of WSTS, this implies $\astate_2\leq \astatep_2$.
Since $\aword\concat\awordpp\in\alang$ and $\alang=\langof{\anults}$, we get $\astate_2\in\finalstates$.
Since $\finalstates$ is upward closed, also $\astatep_2\in\finalstates$.
Hence, $\awordp\concat\awordpp\in\langof{\anults}=\alang$ as desired.

For the if-direction, consider a language $\alang\subseteq\analph^*$ whose Nerode quasi order $\leqof{\alang}$ is a WQO.
We define the \dwsts\ $\anults_{\alang}=
(\analph^{*}, \nerodeleqof{\alang}, \varepsilon, \analph, \transitions, \alang)$.
The states are all words ordered by the Nerode quasi order.
The empty word is the initial state, the language $\alang$ is the set of final states.
Note that $\alang$ is upward closed wrt.\ $\nerodeleqof{\alang}$.
The transition relation is defined as expected, $\transitions(\aword,\aletter)
=\aword\concat\aletter$.
It is readily checked that $\langof{\anults_{\alang}}=\alang$.
\end{proof}

The lemma gives a hint on how to construct the witness language $\witnesslang$: The associated Nerode quasi order $\nerodeleqof{\witnesslang}$ should have an infinite antichain (then it cannot be a~WQO).
To obtain such an antichain, remember that $\witnesslang$ will be accepted by a
WSTS over the Rado WQO~$(\radoset, \radoleq)$~\cite{Rado54}.
It is known that $(\downpowof{\radoset}, \subseteq)$ has an infinite antichain.
Our strategy for the definition of $\witnesslang$ will therefore be to translate the infinite antichain in $(\downpowof{\radoset}, \subseteq)$ into an infinite antichain in $(\analph^*, \nerodeleqof{\witnesslang})$.
We turn to the details, starting with the Rado WQO.
\subsection{Witness Language}

\mbox{}

\paragraph*{Rado Order}
Our presentation of the Rado WQO~\cite{Rado54} follows \cite{Milner1985}. 
The \emph{Rado set} is the upper diagonal, $\radoset = \setcond{(c, r)}{c< r}\subseteq\nat^2$.
The Rado WQO $\radoleq\ \subseteq \radoset\times\radoset$ is defined by:
\begin{align*}
(c_1, r_1)\radoleq(c_2, r_2), \qquad \text{if}\qquad r_1\leq c_2\,\,\vee\,\, (c_1=c_2\,\,\wedge\,\, r_1\leq r_2)\ .
\end{align*}
Given an element $(c, r)$, we call $c$ the \emph{column} and $r$ the \emph{row}, as suggested by Figure~\ref{Figure:RadoComparable}~(left).
Columns will play an important role and we denote column~$i$ by $\radocol{i}=\setcond{(i, r)}{i < r}\subseteq\radoset$.
To arrive at a larger element in the Rado WQO, one can increase the row while remaining in the same column,
or move to the rightmost column of the current row, and select an element to the right, Figure~\ref{Figure:RadoComparable}~(middle).
\begin{figure}[!h]
    \centering
    \scalebox{0.5} {
        \rowsandcolsrado{3}{5}{8}
    }\qquad
    \scalebox{0.5} {
        \markedradofigure{3}{5}{8}
    }\qquad
    \scalebox{0.5} {
        \downclscolrado{3}{8}
    }
    \caption{Rado order with the column and row of $(3,5)$ marked (left), with the elements larger than~$(3,5)$ marked (middle), and with the downward closure of column $3$ marked (right).
    \label{Figure:RadoComparable}}
\end{figure}

It is not difficult to see that $(\radoset, \radoleq)$ is a WQO~\cite{Rado54}.
In an infinite sequence, either the columns eventually plateau out, in which case the rows lead to comparable elements, or the columns grow unboundedly, in which case they eventually exceed the row in the initial pair.
The interest in the Rado WQO is that the WQO property is lost when moving to $(\downpowof{\radoset}, \subseteq)$.
%
%
This failure is due to the following well-known fact.

\begin{lem}[\cite{IdealDecomp}, Proposition 4.2]\label{Lemma:RadoUpAntichain}
$\setcond{\downcls{\radocol{i}}}{i\in\nat}$ is an infinite antichain in $(\downpowof{\radoset}, \subseteq)$.
\end{lem}
\noindent To explain the result, we illustrate the downward closure of a column in Figure~\ref{Figure:RadoComparable}~(right).
Inclusion fails to be a WQO as each column $\radocol{i}$ forms an infinite set that the downward closure $\downcls{\radocol{j}}$ with $j>i$ cannot cover.
Indeed, $\downcls{\radocol{j}}$ only has the triangle to the bottom-left of column $\radocol{j}$ available to cover $\radocol{i}$, and the triangle is a finite set.
We will use exactly this difference between infinite and finite sets in our witness language.
\paragraph*{Definition of $\witnesslang$}
The witness language is the language accepted by $\anults_{\radoset}=(\radoset, \radoleq, \radocol{0},  \analph,
\transitions, \radoset)$.
The set of states is the Rado set, the set of initial states is the first column, and the set of final states is again the entire Rado set.
The latter means that a word is accepted as long as it admits a run.
The letters in
$\analph=\set{\opendyck, \closedyck, \fauxzero}$
reflect the operation that the transitions
$\transitions\subseteq
\radoset\times\analph\times\radoset$ perform on the states:
\begin{alignat*}{5}
       \transitions((c, r), \opendyck)\ &=\ {(c+1,r+1)}
    \qquad \qquad& \transitions((c+1, r+1), \closedyck)\ &=\ {(c, r)}\\
    %
    \transitions((c+1, r), \fauxzero)\ &= \
    {(0, c)}
    \quad&
    \transitions((0, r+1), \fauxzero)\ &=\
    {(0, r)}\ .
    %
\end{alignat*}
We explain the transitions in a moment, but remark that they are designed in a way that makes $\radoleq$ a simulation relation and hence $\anults_{\radoset}$ a $\wsts$.
\begin{lem}\label{Lemma:WSTSMembership}
$\witnesslang\in\langof{\wsts}$.
\end{lem}

To develop an intuition to the witness language, consider
\[
    \witnesslang\cap\opendyck^{*}\concat\closedyck^{*}\concat\fauxzero^{*}\ =\ \setcond{\opendyck^{n}\concat\closedyck^{n}\concat\fauxzero^{i}}{n, i\in\nat}\cup
    \setcond{\opendyck^{n}\concat\closedyck^{k}\concat\fauxzero^{i}}{n, k, i\in\nat, n-k > i}\ .
\]
Until reading the first $\fauxzero$ symbol, the language keeps track of the (Dyck) balance of $\opendyck$ and~$\closedyck$ symbols in a word.
If the balance becomes negative, the word is rejected.
If the balance is non-negative, it is the task of the $\fauxzero$ symbols to distinguish a balance of exactly zero from a positive balance.
Words with a balance of exactly zero get accepted regardless of how many $\fauxzero$ symbols follow.
Word that have a positive balance of $c>0$ when reading the first $\fauxzero$ get rejected after reading $c$-many $\fauxzero$ symbols.
As we show, this is enough to distinguish words with a balance of $c>0$ from words with a balance of $d>0$ for $d\neq c$, and thus obtain infinitely many classes in the Nerode quasi order.
We turn to the details.
\begin{prop}\label{Proposition:Witness}
    $\witnesslang\notin \langof{\updwsts}$
\end{prop}

To prove $\witnesslang\notin\langof{\dwsts}$, we associate with each column $\radocol{i}$ in the Rado WQO the \emph{column language} $\radolangof{i}=\setcond{\aword\in\analph^*}{\apply{\radocol{0}}{\aword}=\radocol{i}}$.
It consists of those words that reach \emph{all} states in~$\radocol{i}$ from the initial column $\radocol{0}$.
The column languages are non-empty.
\begin{lem}\label{Lemma:NonEmpty}
$\radolangof{i}\neq\emptyset$ for all $i\in\nat$.
\end{lem}
\noindent We start from the entire initial column, meaning $\varepsilon\in \radolangof{0}$.
The transitions labeled by $\opendyck$ move from all states in one column to all states in the next column, $\radolangof{i}.\opendyck\subseteq\radolangof{i+1}$.
This already proves the lemma.
The $\closedyck$-labeled transitions undo the effect of the $\opendyck$-labeled transitions and decrement the column, $\radolangof{i+1}.\closedyck\subseteq\radolangof{i}$.
In the initial column, this is impossible, $\apply{\radocol{0}}{\closedyck}=\emptyset$.
We illustrate the behaviour of $\opendyck$ and $\closedyck$ in Figure~\ref{Figure:TransitionEffect} (left).

By Lemma~\ref{Lemma:RadoUpAntichain}, the columns form an antichain in $(\downpowof{\radoset}, \subseteq)$.
The languages $\radolangof{i}$ translate this antichain into (actually several) antichains of the form we need.
When combined, the Lemmas~\ref{Lemma:NonEmpty},~\ref{Lemma:ColumnWordAntichain}, and~\ref{Lemma:WQONerode} conclude the proof of Proposition~\ref{Proposition:Witness}, and therefore Theorem~\ref{Theorem:WSTSneqDWSTS}.
\begin{lem}\label{Lemma:ColumnWordAntichain}
    Every set $\alangp\subseteq\analph^{*}$ with $|\alangp\cap\radolangof{i}|=1$ for all $i\in\nat$ is an antichain in $(\analph^{*}, \nerodeleqof{\witnesslang})$.
\end{lem}

In the rest of the section, we prove Lemma~\ref{Lemma:ColumnWordAntichain}.
The lemma claims that entire column languages are incomparable in the Nerode quasi order, so we write $\alang\nerodeincompof{\witnesslang}\alangp$ if for all $\aword\in\alang$ and all $\awordp\in\alangp$ we have $\aword\not\nerodeleqof{\witnesslang}\awordp$ and $\awordp\not\nerodeleqof{\witnesslang}\aword$.
Difficult is the incomparability with $\radolangof{0}$ stated in the next lemma. The proof will make formal the idea behind the $\fauxzero$-labeled transitions. 
\begin{figure}[h]
    \centering
    \scalebox{0.5} {
        \dycktransitionsrado{3}{7}
    }\qquad\qquad
    \scalebox{0.5} {
        \dominatezerorado{3}{7}
    }
    \caption{The effect of $\opendyck$ and $\closedyck$-labeled transitions on column $3$ (left) and the effect of $\fauxzero$-labeled transitions on columns $0$ and $3$ (right).}
    \label{Figure:TransitionEffect}
\end{figure}
\begin{lem}\label{Lemma:BaseCase}
$\radolangof{0}\nerodeincompof{\witnesslang}\radolangof{k}$ for all $k>0$.
\end{lem}
\begin{proof}
Let $\aword\in\radolangof{0}$ and $\awordp\in\radolangof{k}$, meaning $\aword$ leads to all states in column $0$ while $\awordp$ leads to all states in column $k>0$.
It is easy to find a suffix that shows  $\awordp\not\nerodeleqof{\witnesslang}\aword$, namely $\closedyck$.
Appending $\closedyck$ to $\awordp$ leads to column $\radocol{k-1}$, and so $\awordp\concat\closedyck\in\witnesslang$, while there is no transition on $\closedyck$ from $\radocol{0}$, and so $\aword\concat\closedyck\notin \witnesslang$.

For $\aword\not\nerodeleqof{\witnesslang}\awordp$, we need the $\fauxzero$ transitions.
The idea is to make them fail in $\radocol{k}$ for~$k>0$, and have no effect in $\radocol{0}$.
%
%
The problem is that the states in $\radocol{k}$ must simulate $(0, r)$ for~$r\leq k$.
The trick is to fail with a delay.
%
Instead of having no effect in~$\radocol{0}$, we let the $\fauxzero$ transitions decrement the row.
Instead of failing in $\radocol{k}$, we let the $\fauxzero$ transitions imitate the behavior from $(0, k)$
and move to $(0, k-1)$.
%
%
This is illustrated in Figure~\ref{Figure:TransitionEffect}~(right).

By working with column languages, the $\fauxzero$ transitions fail in $\radocol{k}$ with a delay as follows.
We have $\radolangof{0}.\fauxzero\subseteq\radolangof{0}$ but $\radolangof{k}.\fauxzero\not\subseteq\radolangof{0}$, meaning from $\radocol{0}$ we again reach the entire column~$\radocol{0}$, while from $\radocol{k}$ we only reach the state $(0, k-1)$. 
%
%
%
The decrement behavior in the initial column allows us to distinguish the cases by exhausting the rows.
%
%
%
%
%
Certainly, $\fauxzero^{k-1}$ is enabled in large enough states of~$\radocol{0}$, meaning $\aword\concat\fauxzero^{k}\in\witnesslang$.
The state $(0, k-1)$ reached by $\awordp.\fauxzero$, however, does not enable corresponding transitions, $\awordp\concat\fauxzero^{k}\notin\witnesslang$.
\end{proof}

It is interesting to note that, when executed in $\radocol{k}$ with $k>0$, the $\fauxzero$ transitions resemble reset transitions~\cite{DufourdFinkelReset}.
An analogue of leaving $\radocol{0}$ unchanged despite decrements, however, does not seem to exist in the classical model.
Moreover, reset nets are defined over~$\nat^k$, which is an $\omega^2$-WQO, as opposed to the Rado set.

To conclude the proof of Lemma~\ref{Lemma:ColumnWordAntichain}, we lift the previous result to arbitrary column languages.

\begin{lem}\label{Lemma:MovingColumns}
$\radolangof{i}\nerodeincompof{\witnesslang}\radolangof{j}$ for all $i\neq j$.
\end{lem}
\begin{proof}
Let $i<j$ and consider $\aword\in\radolangof{i}$ and $\awordp\in\radolangof{j}$.
For $\awordp\not\nerodeleqof{\witnesslang}\aword$, we append $\closedyck^j$, which is possible only from the larger column: $\awordp\concat\closedyck^{j}\in\witnesslang$ but $\aword\concat\closedyck^{j}\notin\witnesslang$.
For $\aword\not\nerodeleqof{\witnesslang}\awordp$, we append $\closedyck^i$.
Then $\aword\concat\closedyck^i\in \radolangof{0}$ while $\awordp\concat\closedyck^i\in\radolangof{k}$ with $k>0$.
Now Lemma~\ref{Lemma:BaseCase} applies and yields a suffix $\awordpp$ so that $\aword\concat\closedyck^i\concat \awordpp\in \witnesslang$ but $\awordp\concat\closedyck^i\concat \awordpp\notin \witnesslang$.
\end{proof}

The \wsts\ accepting the witness language $\witnesslang$ uses non-determinism only in the choice of the initial state.
The transitions are deterministic.
Moreover, the Rado WQO is embedded in every non-$\omega^2$-WQO~\cite[Section 2]{JANCAR1999}.
Given the determinizability results in Theorem~\ref{Theorem:PositiveExpressivity}, language $\witnesslang$ thus shows non-determinizability of \wsts\ with minimal requirements.

\section{A Specialized Invariant Construction for \dwsts}\label{Section:SimpleRegSep}
With the determinizability results at hand, we now give an alternative proof of our regular separability result in Theorem~\ref{Theorem:RegSep} for the special case that one of the \wsts\ is deterministic. 
We again use the proof principle for separability developed in Theorem~\ref{thm:core}, which means our alternative proof takes the form of a specialized invariant construction.
The advantage of this invariant construction over the general case is that (i) it merely decomposes an invariant of the product WSTS rather than adding limits of converging sequences and (ii) it works with better known concepts, namely ideals rather than closed sets. 
\begin{thm}\label{Theorem:RegSepLimited}
Let $\anults_1, \anults_2$ be WSTS with $\anults_2$ deterministic and  $\langof{\anults_1}\cap\langof{\anults_2}=\emptyset$. 
Then there is a regular separator of $\langof{\anults_1}$ and $\langof{\anults_2}$ whose states are ideals of the product~WQO. 
\end{thm}

The key to our construction is the notion of ideal completions \cite{InfBranching14,FG09}.
We show that every invariant for a \WSTS yields a
finitely-represented invariant for the corresponding ideal completion.
Theorem~\ref{Theorem:RegSepLimited} follows from this. 
\subsection{Ideal Completion}
An \emph{ideal} in a  \wqo $(\states, \leq)$ is a downward-closed set $\emptyset\neq J\subseteq \states$ which is directed:
for every $j, j'\in J$ there is $j'' \in J$ with $j  \leq j''$ and $j' \leq j''$.
In a product WQO $(\states_1\times \states_2, \leq_{\times})$, the ideals are precisely the products of the ideals in $\states_1$ and~$\states_2$.
\begin{lem}[\cite{KP92,FG09,LS15}]
\label{lem:prodideals}
    A set $J\subseteq \states_1\times \states_2$ is an ideal if and only if $J = J_1 \times J_2$, where $J_1 \subseteq \states_1$
    and $J_2\subseteq \states_2$ are ideals.
\end{lem}
Every downward-closed set decomposes into finitely many ideals.
In fact, the finite antichain property is sufficient and necessary for this.
\begin{lem}[\cite{KP92,FG09,LS15}]
\label{lem:idec}
    In a \wqo, every downward-closed set is a finite union of ideals.
\end{lem}
\noindent
We use $\idec{\states}{Z}$ to denote the set of inclusion-maximal ideals in the downward-closed set~$Z$.
By the above lemma, $\idec{\states}{Z}$ is always finite and
\begin{align}
\label{eq:union}
    Z\ =\ \bigcup \idec{\states}{Z}\ .
\end{align}

We will also make use of the fact that ideals are irreducible in the following sense.

\begin{lem}[\cite{KP92,FG09,LS15}]
\label{lem:idealsirred}
    Let $(\states, \leq)$ be a \wqo.
    If $Z \subseteq \states$ is downward closed and $I \subseteq Z$ an ideal, then
    $I \subseteq J$ for some $J \in \idec{\states}{Z}$.
\end{lem}

The \emph{ideal completion} $(\idcompl{\states}, \subseteq)$ of $(\states, \leq)$ has as elements all ideals
in $\states$.
The order is inclusion.
The ideal completion $\idcompl{\states}$ can be seen as extension of $\states$: every element $\astate\in \states$ is represented
by $\downclosure{\{\astate\}} \in \idcompl{\states}$, and inclusion among such representations coincides with the original quasi order~$\leq$.
Later, we will also need general ideals that may not be the downward closure of a single element.

In \cite{FG09,InfBranching14}, the notion has been lifted to \WSTS\ $\anults=(\states, \leq, \initstates, \analph, \transitions, \finalstates)$.
The ideal completion is the \ULTS 
$\idcompl{\anults}\ =\ (\idcompl{\states}, \subseteq, \idcompl{\initstates}, \analph, \idcompl{\transitions}, \idcompl{\finalstates})$. 
The WQO is replaced by its ideal completion.
The initial states are the ideals in the decomposition of $\downclosure{I}$, 
$\idcompl{\initstates} = \idec{S}{\downclosure{I}}$. 
The transition relation is defined similarly, by decomposing $\downclosure{\succe{\anults}{J}{a}}$ with $J$ an ideal:
\begin{align*}
\idcompl{\transitions}(J, a) \ = \ \idec{S}{\downclosure{\succe{\anults}{J}{a}}}\ .
\end{align*}
The final states are the ideals that intersect $F$, 
    $\idcompl{\finalstates} = \setof{J\in\idcompl{\states}}{J\cap F \neq \emptyset}$. 

The ideal completion yields a ULTS, preserves the language, preserves determinism, and is compatible with forming products.
\begin{lem}
\label{lem:idcompl}
Let $\anults$ be a WSTS.
Then $\idcompl{\anults}$ is a ULTS, we have $\langof{\idcompl{\anults}}=\langof{\anults}$, and if $\anults$ is deterministic, so is $\idcompl{\anults}$.   
Finally, $\idcompl{\anults} \times \idcompl{\anultsp}=\idcompl{(\anults \times \anultsp)}$, with $\anultsp$ another WSTS.
\end{lem}

\noindent
As a matter of fact, $\idcompl{\anults}$ is even finitely branching, but we do not need this property.
\subsection{Invariant Construction}
Ideal completions make it easy to find inductive invariants that are finitely represented. 
Assume the given \UWSTS $\anults$ has an inductive invariant $X$, not necessarily finitely represented. 
By the definition of invariants, $X$ is downward closed.
Thus, by Lemma~\ref{lem:idec}, $X$ is a finite union of ideals.
These ideals are states of the ideal completion $\idcompl{\anults}$. 
To turn $\idec{\states}{X}$ into an inductive invariant of $\idcompl{\anults}$, it remains to take the downward closure of the set.
As the order among ideals is inclusion, this does not add states.
In short, an inductive invariant for $\anults$ induces a finitely-represented inductive invariant for $\idcompl{\anults}$.

\begin{prop}
\label{prop:inducedinvariant}
    If $X \subseteq \states$ is an inductive invariant of $\anults$,
    $\downclosure{\idec{\states}{X}}$ is a finitely-represented inductive invariant of $\idcompl{\anults}$.
\end{prop}
\begin{proof}
    Define $\sepstates = \idec{\states}{X}$.
    Since $\sepstates$ contains all ideals $J \subseteq X$ that are maximal wrt.~inclusion, $\downclosure{\sepstates}$ contains all ideals $J \subseteq X$.
    We observe that
    \begin{align*}
    X\ \stackrel{\eqref{eq:union}}{=}\ \bigcup \sepstates\ =\ \bigcup \downclosure{\sepstates}
    \ .
    \end{align*}
    By Lemma~\ref{lem:idec}, $\sepstates$ is finite and thus $\downclosure{Q}$  finitely represented.
    It remains to check that $\downclosure{\sepstates}$ satisfies the Properties~\eqref{eq:iiI}, \eqref{eq:iiF}, and~\eqref{eq:iiSucc}.

    For Property~\eqref{eq:iiI}, we need to prove $\idec{\states}{\downclosure{\initstates}} \subseteq \downclosure{\sepstates}$.
    Since $X$ is an invariant of~$\anults$, we have $I \subseteq X$. 
    Since $X$ is downward closed, we obtain $\downclosure{I} \subseteq X$.
    Consequently, every ideal that is a subset of $\downclosure{I}$ is also a subset of $X$, and $\downclosure{Q}$ contains all such ideals.
    For Property~\eqref{eq:iiF}, assume towards a contradiction that $\downclosure{Q}$ contains an ideal $J$ that is final in~$\idcompl{\anults}$.
    Then $J$ contains a final state.
    Since $J \subseteq X$, we obtain a contradiction to $X \cap F = \emptyset$, which again holds by the fact that $X$ is an invariant.
    To check the inclusion $\succe{\idcompl{\anults}}{\downclosure{\sepstates}}{a} \subseteq \downclosure{\sepstates}$, we pick an ideal $J \in \downclosure{\sepstates}$ and show $\succe{\idcompl{\anults}}{J}{a} \subseteq \downclosure{\sepstates}$.
    Recall the definition $\succe{\idcompl{\anults}}{J}{a} = \idec{\states}{\downclosure{\succe{\anults}{J}{a}}}$.
    Thus, every element of $\succe{\idcompl{\anults}}{J}{a}$ is an ideal that is a subset of $\downclosure{\succe{\anults}{J}{a}}$.
    We have $\succe{\anults}{X}{a} \subseteq X$ as $X$ is an invariant.
    Since $J \subseteq X$, this implies $\succe{\anults}{J}{a} \subseteq X$. 
    We even have $\downclosure{\succe{\anults}{J}{a}} \subseteq X$ as $X$ is downward closed.
    Hence, every ideal that is a subset of $\downclosure{\succe{\anults}{J}{a}}$ is also subset of $X$, and thus an element of
    $\downclosure{\sepstates}$.
\end{proof}

\begin{proof}[Proof of Theorem~\ref{Theorem:RegSepLimited}]
Consider WSTS $\anults_1$ and $\anults_2$ with disjoint languages and where $\anults_2$ is deterministic.
Let $(\states_1, \leq_1)$ resp. $(\states_2, \leq_2)$ be the underlying sets of states. 
Using Lemma~\ref{lem:idcompl}, we construct the ideal completions $\idcompl{\anults_1}$ and $\idcompl{\anults_2}$, and note that $\idcompl{\anults_2}$ is still deterministic. 
This allows us to invoke the proof principle for regular separability, Theorem~\ref{thm:core}. 
We argue that the product $\idcompl{\anults_1}\times\idcompl{\anults_2}$ has a finitely-represented inductive invariant and, moreover, the elements representing the invariant are ideals of $\states_1\times \states_2$.
Then Theorem~\ref{Theorem:RegSepLimited} follows. 

To arrive at the invariant, we use the equality $\idcompl{\anults_1}\times\idcompl{\anults_2}=\idcompl{(\anults_1\times \anults_2)}$. 
Since $\langof{\anults_1}\cap\langof{\anults_2}=\emptyset=\langof{\anults_1\times\anults_2}$, the synchronized product of the WSTS has an inductive invariant $X$ by Lemma~\ref{Lemma:Invariants}. 
With Proposition~\ref{prop:inducedinvariant}, we get that $\downclosure{\idec{\states_1\times\states_2}{X}}$ is a finitely-represented inductive invariant of $\idcompl{(\anults_1\times \anults_2)}$. 
The elements representing the invariant are ideals of $\states_1\times\states_2$, as required.
\end{proof}

\section{Downward-Compatible WSTS}\label{Section:DWSTS}
We show that the regular separability and non-determinizability results we have obtained for upward-compatible \wsts\ so far can be lifted to downward-compatible WSTS (\downwsts).
In \downwsts, smaller states simulate larger ones and the set of final states is downward closed.
We lift our results by establishing general relations between the language classes $\langof{\upwsts}$, $\langof{\downwsts}$, $\langof{\updwsts}$, and $\langof{\downdwsts}$.
Figure~\ref{Figure:Relations} summarizes them.
\begin{figure}[h]
    \centering
    \begin{tikzpicture}

        \node (upwsts) at (2, 4) {$\langof{\upwsts}$};
        \node (downwsts) at (7, 4) {$\langof{\downwsts}$};
        \node (detupwsts) at (1, 2) {$\langof{\updwsts}$};
        \node (detdownwsts) at (8, 2) {$\langof{\downdwsts}$};
        \node (omegasqupwsts) at (0, 0) {$\langof{\upomegasqwsts}$};
        \node (omegasqdownwsts) at (9, 0) {$\langof{\downomegasqwsts}$};

        \node (finbdownwsts) at (10.5, 2) {$=\langof{\downfinbranchwsts}$};
        \node (finbdownwstsreason) at (9.5, 2.4) {Proposition~\ref{Proposition:DWSTSfinb}};

        \node (finbupwsts) at (-0.8, 0.37) {$\langof{\upfinbranchwsts}$};

        \path[<->]
        (detupwsts) edge
        node[yshift=-0.75em]{$=_{\cmplabel}$,\ {\text{\footnotesize Lemma~\ref{Lemma:Complementation}}}}
        node[yshift=0.75em]{
            $\not\subseteq_{\revlabel}$, $\not\supseteq_{\revlabel}$,\ {\text{\footnotesize Lemma~\ref{Lemma:DetRevDet}}}}
        (detdownwsts)

        (upwsts) edge node[yshift = 0.5em] {$=_{\revlabel}$,\ {\text{\footnotesize Lemma~\ref{Lemma:Reversal}}}} (downwsts);

        \path[->]
        (detupwsts) edge node[xshift=3.5em] {$\subsetneq,\ {\text{\footnotesize Theorem~\ref{Theorem:WSTSneqDWSTS}}}$} (upwsts)
        (detdownwsts) edge node[xshift=-3.5em] {$\subsetneq,\ {\text{\footnotesize  Theorem~\ref{Theorem:ComplWitnessLanguage}}}$} (downwsts);

        \path[<->]
        (omegasqupwsts) edge node[yshift = 0.5em] {$=_{\revlabel}$,\ {\text{\footnotesize Lemma~\ref{Lemma:Reversal}}}} (omegasqdownwsts);

        \path[->]
        (omegasqupwsts) edge node[xshift=3.5em] {$\subseteq,\ {\text{\footnotesize Theorem~\ref{Theorem:PositiveExpressivity}}}$} (detupwsts)
        (omegasqdownwsts) edge node[xshift=-3.5em] {$\subseteq,\ {\text{\footnotesize  Theorem~\ref{Theorem:ComplWitnessLanguage}}}$} (detdownwsts)
        (finbupwsts) edge node[xshift=3.5em] {} (detupwsts);
    \end{tikzpicture}
    \caption{Relations between language classes.\label{Figure:Relations}}
\end{figure}

%
%
%
%
\paragraph*{Downward Compatibility}
A \emph{downward-compatible LTS} (\downlts)~$\adownlts=(\states, \initstates, \analph, \transitions, \finalstates)$ is an LTS whose states are equipped with a quasi order $\leq\ \subseteq \states\times \states$ so that the following holds.
The final states are downward closed, $\downcls{\finalstates}=\finalstates$, and $\geq$ is a simulation relation.
Recall that this means for all $\astate_1\leq\astatep_1$  and for all $\astatep_2\in\transitions(\astatep_1, \aletter)$ there is $\astate_2\in\transitions(\astate_1, \aletter)$ with $\astate_2\leq\astatep_2$. We denote the class of deterministic \downlts\ by \downdlts.
We use $\langof{\downlts}$ and $\langof{\downdlts}$ to refer to the classes of all \downlts\ resp. $\downdlts$ languages.
If $\leq$ is also a WQO, we call $\adownlts$ a \emph{downward-compatible WSTS} ($\downwsts$). 
\subsection{Relations between $\langof{\downlts}$ and $\langof{\ults}$}
%
%
The languages accepted by \downlts\ are the reverse of the languages accepted by \ults, and vice-versa.
This is easy to see by reversing the transitions.
Let $\anults=(\states, \leq, \initstates, \analph, \transitions, \finalstates)$
be an \ults.
We define $\reverseof{\anults}=(\states, \leq, \finalstates, \analph, \reverseof{\transitions}, \downcls{\initstates})$ to be its reversal.
The initial and final states are swapped and the direction of the transitions is flipped, $\reverseof{\transitions}=\setcond{(\astate, \aletter, \astatep)}{(\astatep, \aletter, \astate')\in\transitions, \astate\leq\astate'}$.
Note that we close the initial states downwards and add transitions from states smaller than the original target.
This corresponds to the assumption that the original transitions relate downward-closed sets.
%
%
The construction can also be applied in reverse to get a ULTS $\reverseof{\adownlts}$ from a DLTS $\adownlts$.
\begin{lem}\label{Lemma:Reversal}
If $\anults\in \ults$ (\wsts, $\omegasqwsts$), then $\reverseof{\anults}\in\downlts$ (\downwsts, $\downomegasqwsts$) and $\langof{\reverseof{\anults}}=\reverseof{\langof{\anults}}$.
If $\adownlts\in \downlts$ (\downwsts, $\downomegasqwsts$), then $\reverseof{\adownlts}\in\ults$ (\wsts, $\omegasqwsts$) and $\langof{\reverseof{\adownlts}}=\reverseof{\langof{\adownlts}}$.
\end{lem}

The \downdlts\ languages are precisely the complements of the \udlts\ languages.
For a \udlts\ or \downdlts\ $\anults=(\states, \leq, \initstate, \analph, \transitions, \finalstates)$, we define the complement $\complementof{\anults}=(\states, \leq, \initstate, \analph, \transitions, \complementof{\finalstates})$ by complementing the set of final states~\cite[Theorem 5]{RabinScott59}.
\begin{lem}\label{Lemma:Complementation}
$\anults\in\udlts$ $(\updwsts)$ iff $\complementof{\anults}\in\downdlts$ $(\downdwsts)$,  and $\langof{\complementof{\anults}}=\complementof{\langof{\anults}}$.
\end{lem}
Behind this is the observation that, under determinism, $\leq$ is a simulation if and only if $\geq$ is~\cite[Theorem 3.3(ii)]{M71}.
\subsection{Lifting Results}\label{Subsection:GeneralizingResults}

\mbox{}

\paragraph*{Regular Separability of \downwsts}
We obtain the regular separability of disjoint \downwsts\ languages as a consequence of the previous results.
More precisely, we need Lemma~\ref{Lemma:Reversal}, Theorem~\ref{Theorem:RegSep}, and the closure of the regular languages under reversal.
\begin{thm}
    Let $\alang_1, \alang_2\in\langof{\downwsts}$.
    We have $\alang_1\regsep\alang_2$ if and only if $\alang_1\cap\alang_2=\emptyset$.
\end{thm}
\paragraph*{Non-Determinizability of \downwsts}
To show that \downwsts\ cannot be determinized, recall our witness language $\witnesslang$ from Section~\ref{Section:NonDet}.
Surprisingly, we have the following.
\begin{lem}\label{Lemma:WitnessReverse}
$\reverseof{\witnesslang}\in\langof{\downdwsts}$ and $\reverseof{\complementof{\witnesslang}}\in \langof{\updwsts}$.
\end{lem}
For the first claim, recall that the witness language is accepted by the \wsts\ $\anults_{\radoset}$.
The \downwsts\ $\reverseof{\anults_{\radoset}}$ has one minimal
initial state, and transition images $\reverseof{\transitions}(\astate, \aletterp)$ with one minimal element for all $\astate\in\radoset$ and $\aletterp\in\analph$.
Removing simulated states yields a deterministic \downwsts.
The details are in the full version \cite{keskinmeyer2023sep}.
Actually, the idea of removing simulated states reappears in a moment when we study finitely branching \downwsts.
For the second claim, $\complementof{\reverseof{\witnesslang}}\in\langof{\updwsts}$ by Lemma~\ref{Lemma:Complementation}.
But $\complementof{\reverseof{\witnesslang}}=\reverseof{\complementof{\witnesslang}}$, and so $\reverseof{\complementof{\witnesslang}}\in \langof{\updwsts}$.
Behind this is the fact that bijections commute with complements, and reversal is a bijection.

The lemma allows us to prove non-determinizability for \downwsts.
Notably, we do not need a characterization for the languages of deterministic \downwsts.
\begin{thm}\label{Theorem:ComplWitnessLanguage}
    $\complementof{\witnesslang}\in\langof{\downwsts}\setminus\langof{\downdwsts}$ and so $\langof{\downwsts}\neq\langof{\downdwsts}$.
\end{thm}

\begin{proof}
By Lemma~\ref{Lemma:WitnessReverse}, $\reverseof{\complementof{\witnesslang}}\in \langof{\updwsts}$.
Lemma~\ref{Lemma:Reversal} yields $\complementof{\witnesslang}\in\langof{\downwsts}$.
%
%
Towards a contradiction, suppose $\complementof{\witnesslang}\in\langof{\downdwsts}$.
Then $\witnesslang\in\langof{\updwsts}$ by Lemma~\ref{Lemma:Complementation}.
This contradicts Proposition~\ref{Proposition:Witness}.
\end{proof}

In Section~\ref{Section:NonDet}, we have studied restricted classes of WSTS and shown that they admit determinizability results.
We now consider the same restrictions for \downwsts.
We proceed with $\downomegasqwsts$, $\downwsts$ whose states are ordered by an $\omega^{2}$-WQO.
The same argument we used for $\upomegasqwsts$ shows that $\downomegasqwsts$ determinize.
\begin{prop}
    $\langof{\downomegasqwsts}\subseteq\langof{\downdwsts}$
\end{prop}
We move on to finitely branching \downwsts.
In contrast to \upwsts, finite branching does not restrict the expressiveness of \downwsts.
This is because any \downwsts\ can be made into a finitely branching one as follows.
Let $\adownlts=(\states, \leq, \initstates, \analph, \transitions, \finalstates)$ be a \downwsts.
In a WQO, every set $\emptyset\neq X\subseteq\states$ has only finitely many minimal elements.
For each symbol $a\in\analph$, and state $\astate\in\states$, we remove all transitions $(\astate, a, \astatep)$, where $\astatep$ is not minimal in $\apply{\astate}{a}$.
Thanks to downward compatibility, this does not change the language.
\begin{prop}\label{Proposition:DWSTSfinb}
    $\langof{\downfinbranchwsts}=\langof{\downwsts}$
\end{prop}

\subsection{Consequences}\label{Subsection:RelationsBetweenLanguages}
We have shown that neither upward- nor downward-compatible \wsts\ can be determinized.
This does not yet rule out the possibility of determinizing an upward-compatible \wsts\ into a downward-compatible one, and vice versa.
In the light of Lemma~\ref{Lemma:Reversal}, we should allow the determinization to reverse the language.
We now show that also this form of reverse-determinization is impossible: there are even deterministic languages that cannot be reverse-determinized.
This is by Lemma~\ref{Lemma:WitnessReverse}, Proposition~\ref{Proposition:Witness}, and Theorem~\ref{Theorem:ComplWitnessLanguage}. 
\begin{lem}\label{Lemma:DetRevDet}
$\reverseof{\witnesslang}\in\langof{\downdwsts}$ but $\witnesslang\notin\langof{\updwsts}$.
Similarly, $\reverseof{\complementof{\witnesslang}}\in\langof{\updwsts}$ but $\complementof{\witnesslang}\notin\langof{\downdwsts}$
\end{lem}

After reversal, both witness languages $\witnesslang$ and $\complementof{\witnesslang}$ can be accepted by a deterministic~WSTS.
When it comes to separability, this means already the results from \cite{CONCUR18}, given here as Theorem~\ref{Theorem:RegSepLimited}, apply to them.
A consequence of Lemma~\ref{Lemma:DetRevDet}, however, is that there are \wsts\ languages that can neither be determinized nor reverse-determinized.
An instance is the language $K=\witnesslang\concat\#\concat\reverseof{\complementof{\witnesslang}}$ with $\#$ a fresh letter.
\begin{lem}\label{Lemma:BidirectionalNonDet}
$K\in\langof{\upwsts}$, $K\not\in\langof{\updwsts}$, and $\reverseof{K}\not\in\langof{\downdwsts}$.
\end{lem}
When considering disjoint $\alangp_1, \alangp_2\in\langof{\wsts}$ that can neither be determinized nor reverse-determinized, Theorem~\ref{Theorem:RegSepLimited}, which restricts one \WSTS to be deterministic, does not apply.
Theorem~\ref{Theorem:RegSep} is stronger and yields $\alangp_1\regsep \alangp_2$.
The situation is similar for downward-compatible WSTS.

\section{Conclusion and Future Work}\label{Section:Conclusion}
We have shown that disjoint WSTS languages are always separated by a regular language.
We have also shown that deterministic WSTS accept a strictly weaker class of languages than their non-deterministic counterparts.
We complemented these result by showing that determinization is possible under mild assumptions ($\omegasqwsts$ and finitely branching $\upwsts$), and extending the results to downward-compatible \wsts.

As for future work, it would be interesting to develop new verification algorithms based on regular separability results.
Consider the safety verification problem for concurrent programs, formulated as (non-)coverability in Petri nets, as an example. 
Our separability result suggests the following compositional approach to proving safety: we decompose the given Petri net into two components (whose product is the original net), over-approximate the components by regular languages, and show the disjointness of these regular languages.
The benefit of our separability result is that it should give a completeness guarantee, if the over-approximation is done well or iteratively refined.

On the theoretical side, it would be desirable to have a better understanding of regular separability: when is separability guaranteed to hold, when is the algorithmic problem of checking separability decidable, and what is the size of separators?
Behind our separability result is a close link between regular separability and finitely-represented inductive invariants.
Does such a link exist in other separability results as well, in particular the one for VASS reachability languages~\cite{KM24}?
Is there a link between invariants and the basic separators introduced in~\cite{CZ20}? 
Can our invariant closure be understood in a topological way in order to generalize the result~\cite{GL07,GL10}.
As a concrete generalization, does our result continue to hold for well-behaved transition systems~\cite{BlondinFM16}, ULTS that only satisfy the finite antichain property?
Also interesting would be a generalization to weighted languages~\cite{Weighted}.
There are first studies of separability for languages of infinite words~\cite{BMZ23,BKMZ24} and trees~\cite{Goubault-Larrecq16}, and we believe these settings deserve more attention. 

It would also be interesting to study instead of regular separability the separability by subregular languages. 
Prominent classes of subregular language are the languages definable in first-order logic, in first-order logic with two variables, the locally testable languages, and the locally threshold-testable languages. 
For all of them, the separability problem is understood in the case that the input languages are regular~\cite{PlaceRZ13mfcs,PlaceZ14lics,LTSepPvRZ13}.
However, for the languages of infinite state systems, like the $\Z$-VASS languages, $1$-VASS languages, VASS coverability languages, or VASS reachability languages,
the corresponding subregular separability problems seem to be highly challenging (though their regular separability counterparts are decidable~\cite{CCLP17,CL17,CONCUR18,CZ20,CONCUR23,KM24}).
The reason is that these problems call for a combination of techniques from infinite state systems and algebraic language theory.


\bibliography{bib}
\bibliographystyle{alphaurl}

\end{document}